\documentclass[a4paper, 11pt]{article}
\usepackage[top=2.5cm, bottom=2.5cm, left=2.25cm, right=2.25cm]{geometry}
\usepackage[UKenglish]{babel}

\usepackage[usenames, dvipsnames]{xcolor}

\usepackage[utf8]{inputenc}
\usepackage{amsfonts}
\usepackage{enumitem,csquotes}
\usepackage{amsmath}
\usepackage{hhline}
\usepackage{amsthm}
\usepackage{amssymb}
\usepackage{array,multirow,makecell}

\usepackage{fancybox}

\usepackage{framed}
\usepackage{caption}
\usepackage{bbm}
\usepackage{graphicx}
\usepackage{subcaption}
\usepackage{url}
\usepackage{here}
\usepackage{varioref}
\usepackage{enumitem} 

\usepackage{tikz}
\usetikzlibrary{automata,intersections,shapes,arrows,calc,positioning,decorations}
\usepackage{pgfplots}
\usepackage{tkz-fct}

\newtheorem{theorem}{Theorem}[section]
\newtheorem{proposition}[theorem]{Proposition}
\newtheorem{corollary}[theorem]{Corollary}
\newtheorem{lemma}[theorem]{Lemma}
\theoremstyle{definition}

\theoremstyle{remark}
\newtheorem{remark}[theorem]{Remark}
\newtheorem{example}[theorem]{Example}

\newcolumntype{L}[1]{>{\raggedright\let\newline\\\arraybackslash\hspace{0pt}}m{#1}}
\newcolumntype{C}[1]{>{\centering\let\newline\\\arraybackslash\hspace{0pt}}m{#1}}
\newcolumntype{R}[1]{>{\raggedleft\let\newline\\\arraybackslash\hspace{0pt}}m{#1}}
\newcolumntype{P}[1]{>{\centering\arraybackslash}p{#1}}

\usepackage{mathtools} 
\usepackage{extarrows}

\newcommand{\mC}{\mathbb C}

\newcommand{\mR}{\mathbb R}
\newcommand{\mZ}{\mathbb Z}
\newcommand{\mN}{\mathbb N}

\newcommand{\sB}{{\mathcal B}}

\newcommand{\sO}{{\mathcal O}}

\newcommand{\bG}{{\bf G}}

\newcommand{\bpm}[1]{\begin{pmatrix}#1\end{pmatrix}}

\newcommand{\rd }{{\rm  d}}

\newcommand{\e}{{\rm  e}}
\newcommand{\eps}{\varepsilon}

\newcommand{\mc}{\,;\,}

\newcommand{\diff}[2]{\frac{\rd  #1}{\rd  #2}}

\DeclareMathOperator*{\esssup}{ess\,sup}

\newcommand{\mOne}{\mathbbm{1}}

\numberwithin{equation}{section}
\numberwithin{figure}{section}
\numberwithin{table}{section}

\newcommand{\tm}{{\texttt m}}
\newcommand{\tn}{{\texttt n}}
\newcommand{\tp}{{\texttt p}}

\usepackage[urlcolor = black, colorlinks = true, citecolor=black, linkcolor = black]{hyperref}

\title{A linear dissipativity approach to incremental input-to-state stability for a class of positive Lur'e systems}

\author{Violaine Piengeon\thanks{Department of Mathematics, University of Namur, Rue de Bruxelles 61, B-5000 Namur, Belgium, email: {\tt violaine.piengeon@student.unamur.be}} \and Chris Guiver\thanks{School of Computing, Engineering \& the Built Environment, Edinburgh Napier University, Merchiston Campus, Edinburgh, EH10 5DT, UK, email: {\tt c.guiver@napier.ac.uk}}\,
\thanks{Corresponding author}}
\date{Submitted to Journal December 2023, Submitted to Arxiv Feb 2024}	

\makeatletter
\let\thetitle\@title
\let\theauthor\@author
\let\thedate\@date
\makeatother

\usepackage{mathabx}

\begin{document}

\maketitle

\begin{abstract}
    Incremental stability properties are considered for certain systems of forced, nonlinear differential equations with a particular positivity structure. An incremental stability estimate is derived for pairs of input/state/output trajectories of the Lur'e systems under consideration, from which a number of consequences are obtained, including the incremental exponential input-to-state stability property and certain input-output stability concepts with linear gain. Incremental stability estimates provide a basis for an investigation into the response to convergent and (almost) periodic forcing terms, and is treated presently. Our results show that an incremental version of the real Aizerman conjecture is true for positive Lur'e systems when an incremental gain condition is imposed on the nonlinear term, as we describe. Our argumentation is underpinned by linear dissipativity theory --- a property of positive linear control systems. 
\end{abstract}



\section{Introduction}

Consider the following system of nonlinear controlled and observed differential equations  
\begin{equation}\label{eq:lure_4b}
\left.
    \begin{aligned}
        \dot{x}(t) &= Ax(t)+B_1f(t,C_1x(t))+B_2w(t), \\
         y(t)&= 
         \bpm{
            y_1(t)\\ y_2(t)
        }
        = 
        \bpm{
            C_1\\ C_2
        }
        x(t),
    \end{aligned}
    \right \}
\end{equation}
a so-called forced four-block Lur'e (also Lurie or Lurye) system. Here $A, B_i, C_i$ are appropriately-sized matrices, and $f$ is a (time-varying) nonlinear function.  The variables~$x, w$ and~$y$ in~\eqref{eq:lure_4b} denote the state, forcing term, and output, respectively. Depending on the context, the forcing term may represent a noise, disturbance or control variable. The details of system~\eqref{eq:lure_4b} are given later.

Lur'e systems arise as the feedback connection of a linear control system and a nonlinear output feedback, and a block diagram of this feedback connection is shown in Figure~\ref{fig:lure_4b}. Here the input variable to the linear plant comprises a nonlinear output feedback term and a forcing term, and the two components~$y_1$ and~$y_2$ of the measured variable~$y$ are used for feedback and are a performance output, respectively. As such, the term four-block relates to the four mappings from the two input- to the two output- variables in~\eqref{eq:lure_4b}.

%
%
\begin{figure}[h!]%
\centering
\begin{tikzpicture}
		\coordinate (O) at (0,0);
		\node[draw, thick, minimum width=0.5cm, minimum height=0.85cm, anchor=south west, text width = 2.5cm, align = center] (P) at (O) {$(A,B,C)$};
		\node[draw, thick, minimum width=0.5cm, minimum height=0.85cm, text width=1cm, align=center] (C) at ($(P.270) - (0,1)$) {$f$};
		\draw[thick, -latex] ($(P.170) - (1,0)$) -- (P.170) node[above, pos=0.5] {\footnotesize{$w$}};
		\draw[thick,-latex] (P.10) -- ($(P.10)+(1,0)$)node[above, pos=0.4] {\footnotesize{$y_2$}};
		\draw[thick] (P.350) -- ($(P.350)+(0.5,0)$);
		\draw[thick,-latex] ($(P.350)+(0.5,0)$) |- (C.0)node [right, pos=0.3]{\footnotesize{$y_1$}};
		
		
		\draw[thick] (C.180)  -|  ($(P.190)-(0.5,0)$) node [left, pos=0.8]{\footnotesize{$u$}};
		\draw[thick,-latex] ($(P.190)-(0.5,0)$) -- (P.190);
	\end{tikzpicture}
\caption{Block diagram of forced four-block Lur'e system}%
\label{fig:lure_4b}%
\end{figure}
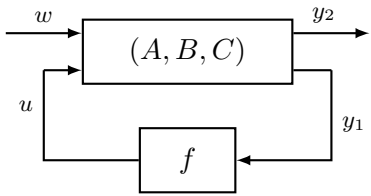%

Lur'e systems are a common and important class of nonlinear control systems, and are consequently well-studied objects. The study of stability properties of Lur'e systems, meaning a variety of possible notions, is called {\em absolute stability theory}, tracing its roots to the 1940s and the Aizerman conjecture~\cite{aizerman1949}, and ordinarily adopts the perspective that the nonlinearity $f$ in~\eqref{eq:lure_4b} is uncertain, being specified rather in terms of given norm or sector bounds. 
An absolute stability criterion is a sufficient condition for stability, usually formulated in terms of the interplay of frequency properties of the linear system and norm or sector bounds of the nonlinearity. Absolute stability theory traditionally comes in one of two strands: the first adopts Lyapunov approaches to deduce global asymptotic stability of unforced (that is, $w=0$) Lur'e systems, as is the approach across~\cite{MR2381711,MR2116013,k02}. The second strand is an input-output approach, pioneered by Sandberg and Zames in the 1960s, to infer $L^2$- and $L^{\infty}$- stability; see, for example,~\cite{MR0490289,MR1946479}. The extensive literature reviewed in~\cite{liberzon2006essays} evidences the interest that absolute stability theory has generated. Classical absolute stability criteria include: the complex  (or generalised) Aizerman conjecture~\cite{hinrichsen1992destabilization}, which is true despite its name; the Circle~\cite{MR2381711,k02} and Popov criteria~\cite{popov1962absolute}; the Integral Quadratic Constraints (IQCs) framework~\cite{megretski1997system}, and; Zames-Falb multipliers. 
We refer the reader to~\cite{carrasco2016zames} for a tutorial on Zames-Falb multipliers, including background on the history and recent perspectives.

%
%
Input-to-State Stability (ISS) and its variants are a suite of stability concepts for controlled (or forced) systems of nonlinear differential equations and, roughly speaking, ensure natural boundedness properties of the state, in terms of (functions of) the initial conditions and inputs. ISS dates back to the 1989 work of Sontag~\cite{sontag1989smooth}, with subsequent developments in the 1990s by Sontag and others across, for example,~\cite{jiang1994small,sontag1995characterizations,sontag1996new}. There is now a vast literature on ISS, with surveys including~\cite{dashkovskiy2011input,sontag2008input}, as well as the more recent text~\cite{MR4592570}. One strength of the input-to-state stability paradigm is that it both encompasses and unifies asymptotic and input-output approaches to stability, and to quote~\cite[Preface]{MR4592570}: {\em ``...revolutionized our view on stabilization of nonlinear systems...''}. 
The practical importance of ISS is that ``small'' external noise, disturbance or unmodelled dynamics result in a correspondingly ``small'' effect on the resulting state. 
%
%
Considerable effort has been denoted to establishing ISS properties for Lur'e systems, originating in the work of Arcak and Teel~\cite{arcak2002input}. These authors initiated a line of enquiry investigating the extent to which classical absolute stability criteria may be strengthened to ensure certain ISS properties of the corresponding forced Lur'e system, with~\cite[Theorem 1]{arcak2002input} being reminiscent of the positivity theorem for absolute stability. The work of Logemann, and his students and collaborators, dating back to~\cite{jayawardhana2009input} and including~\cite{MR3499547,gilmore2020infinite,guiver2020circle,jayawardhana2011circle,sarkans2015input,sarkans2016input} has shown that suitably-strengthened versions of the complex Aizerman conjecture and circle criterion ensure various ISS-type properties. 

%
%
Incremental stability broadly refers to bounding the difference of two arbitrary trajectories of a given system in terms of the difference between initial states and, if included, input terms; see, for instance~\cite{aminzare2014contraction,rwm13,sepulchre2022incremental}. Of course, for linear control systems, incremental stability concepts are equivalent to their corresponding stability concepts via the superposition principle, and this need not be the case for nonlinear systems. Incremental ISS properties have been studied across~\cite{angeli2002lyapunov,angeli2009further} where, amongst a range of results, the somewhat surprising differences between usual ISS properties and their incremental versions are noted. It has recently been commented in~\cite{sepulchre2022incremental} that, roughly speaking, incremental stability concepts are at least as important as stability notions alone and, furthermore, that incremental stability used to have more prominence in the control theory community than perhaps it currently does, and that attention should refocus on this area. One motivation for studying incremental ISS properties, as is well known, is that they provide a toolbox for nonlinear observer design, synchronization-type problems, and the analysis of convergent or (almost) periodic inputs in forced differential equations. Some applications of incremental stability are discussed in~\cite{giaccagli2023further}. Recently, in~\cite{revay2020lipschitz} incremental stability concepts are used to infer robustness properties of certain neural network architectures. Here, we focus on incremental ISS and input-output stability properties of~\eqref{eq:lure_4b}, under the assumption that~\eqref{eq:lure_4b} admits a certain positivity structure which we describe.

%
%
Positive dynamical systems are dynamical systems that leave a positive cone invariant. The academic literature has different conventions over the use of the words ``positive'' and ``nonnegative'' in this context, and they are often used interchangeably. The most natural and obvious positive cone is arguably the nonnegative orthant in real Euclidean space, equipped with the partial ordering of componentwise inequality. There are a myriad of physically motivated examples where the state variables denote necessarily nonnegative quantities, such as abundances, densities, concentrations, flow rates or prices and positive dynamical systems typically arise from mass balance relationships of these quantities --- indeed, quoting~\cite[p. xv]{MR2655815} ``{\em [they] are widespread in biology, chemistry, ecology, economics, genetics, medicine, sociology, and engineering, and play a key role in the understanding of these disciplines.''} The seminal work of Perron and Frobenius in the early 1900s on irreducible and primitive matrices 
underpins the study of positive systems described by linear dynamic equations. A range of nonlinear extensions of these results appear across the literature (see, for example,~\cite{MR3309532}).  There are a number of monographs and textbooks devoted (either in whole or partially) to the study of positive systems, including~\cite{MR1019319,MR1298430,MR2655815,MR3309532}. Positive dynamical systems are closely related to the concept of monotone dynamical systems~\cite{MR2182759,MR1319817}, where the state variables inherit any ordering of the initial states.

Augmenting positive or monotone dynamical systems with input and output variables leads to so-called positive and monotone control systems, see the texts~\cite{MR1784150,MR2655815} or~\cite{angeli2003monotone}, respectively. Again the nomenclature is confusing, as linear positive control systems, often simply called positive systems, are not strictly speaking linear since subtraction cannot always be performed within positive cones. 
 Consequently, there are distinctions between results for positive systems and usual linear control systems, including on reachability and controllability~\cite{MR1839876}, observability~\cite{MR2343058}, realisability~\cite{benvenuti2004tutorial} and stabilisability~\cite{MR2021961}. However, the additional structure associated with positive systems is often helpful and simplifies matters. As an example, the analysis of positive systems benefits from readily constructed classes of scalable Lyapunov functions~\cite{MR3356067}. Writing in 2023, a recent and thorough review of positive systems appears in~\cite{rantzer2021scalable}. 
 Another feature of positive systems, which we utilise presently, is that of so-called {\em linear dissipativity}.


%
%
Dissipativity (or passivity) theory as commonly used in systems and control theory dates back to the seminal work of Willems~\cite{willems1972dissipativeI}, where the notions of supply rate and storage function were introduced, which capture (and generalise) the notion of a system storing and dissipating energy over time. Dissipativity theory is central to control, as consideration of energy is a natural framework for addressing stability. Much attention has been devoted to studying {\em quadratic} supply rates~\cite{willems1972dissipativeII}, as multiple notions of energy (or power) are equal to the product of power-conjugate variables, such as force and velocity in mechanical systems. Two classical notions of quadratic dissipative systems that first arose in circuit theory go by the names of impedance passive and scattering passive, also respectively known as passive and contractive, or bounded real and positive real in the frequency domain. 
Two famous results, sometimes called the Bounded Real Lemma and Positive Real Lemma, provide a complete state-space characterisation of these two notions of dissipativity, respectively, see, for example,~\cite{anderson1973network}. It is known from the work of Haddad {\em et al.}~\cite[Ch. 5]{MR2655815}, there in the context of absolute stability, and that of Briat~\cite{MR3126787} and Rantzer~\cite{MR3356067} in the context of robust control, that positive linear systems may be dissipative with respect to a {\em linear} supply rate and a {\em linear} storage function, that is, are what is termed {\em linearly dissipative}. These ideas seemingly trace their roots back to Ebihara {\em et al.}~\cite{ebihara20111} and have been extended to linear control systems on general positive cones in~\cite{shen2017input}. 
Linear dissipativity captures mass (im)balance, or change in abundance or concentration, in positive systems which may decay, deteriorate or, indeed, dissipate, over time. 

%
%
Presently, our main result is Theorem~\ref{thm:incremental} which shows that a linear dissipativity hypothesis on the underlying linear system, and a certain incremental bound for the nonlinearity $f$, are together sufficient for an incremental stability estimate, which may roughly be interpreted as a weighted mass (im)balance (or dissipation) of differences of input/state/output trajectories of~\eqref{eq:lure_4b}. From this inequality a number of consequences are obtained and presented in Corollary~\ref{cor:incremental_stability} and across Section~\ref{sec:consequences}, where the state- and output- response of~\eqref{eq:lure_4b} to convergent and (almost) periodic forcing terms is investigated. Briefly, the incremental stability estimate we derive is shown to lead to desirable and predictable behaviour in terms of the state- and output-response of~\eqref{eq:lure_4b} to convergent and (almost) periodic forcing terms, in senses we describe. In summary, our application of linear dissipativity theory to positive Lur'e systems with nonlinear terms which satisfy an incremental gain condition to develop incremental input-to-state stability results comprises the main contribution of the present work. Whilst linear dissipativity theory has been used in~\cite{MR2104542} as a basis for Lyapunov stability for positive Lur'e systems (and more generally the feedback connection of positive dynamical systems), ISS concepts do not appear in~\cite{MR2104542}, and neither do incremental stability considerations. Various input-to-state stability results for positive Lur'e inclusions are obtained in~\cite{guiver2019small}, but again there is no focus in that work on incremental stability properties, and so again the overlap with the present work is minimal.

Thus, on the one hand, our present work may be seen as a continuation of~\cite{drummond2022aizerman}, which demonstrates that an Aizerman conjecture for certain positive Lur'e systems is true, to essentially show that the original Aizerman conjecture combined with a linear incremental gain condition on the nonlinear term is sufficient for various incremental stability notions of positive Lur'e systems. On the other hand, the aspects of our work addressing the response to almost periodic forcing terms are inspired by comparable results in~\cite{gilmore2020infinite,gilmore2021incremental}.

%
%
The paper is organised as follows. After gathering requisite preliminaries in Section~\ref{sec:preliminaries}, Section~\ref{sec:lure_systems} introduces the Lur'e systems under consideration, respectively. Our main incremental stability results appear in Section~\ref{sec:incremental_stability}. Section~\ref{sec:consequences} presents consequences of these incremental stability results in terms of the response of the Lur'e systems under consideration to convergent and (almost) periodic forcing terms. Examples are contained in Section~\ref{sec:examples} and summarising remarks appear in Section~\ref{sec:summary}. Some technical material appears in the Appendix. 

%
%
\section{Preliminaries}\label{sec:preliminaries}

\subsection{Notation}
Most notation is standard. As usual, let~$\mN$,~$\mR$,~$\mC$ and~$\mR_+$ denote the set of positive integers, the fields of real and complex numbers and the set of nonnegative real numbers, respectively. For~$\tn, \texttt{m} \in \mN$,~$\mR^{\tn}$ and~$\mR^{\tn \times  \texttt{m}}$ denote~$\tn$-dimensional real Euclidean space and the set of~$\tn \times \texttt{m}$ matrices with real entries, respectively. We equip~$\mR^\tn$ with a (any) vector norm, denoted~$\|\cdot\|$, and use the same symbol to represent the corresponding induced matrix norm. 

\subsection{Matrices, vectors and norms}

Let~$M = (m_{ij})\in \mR^{\tn \times \texttt{m}}$ and~$N = (m_{ij})\in \mR^{\tn \times \texttt{m}}$. With the convention that~$\mR^{\tn \times 1} = \mR^{\tn}$, we write 
\begin{align*}
    M & \geq N\quad\,\mbox{if~$m_{ij}\geq n_{ij}$ for all~$i$ and~$j$},\\
    M &> N\quad\,\mbox{if~$M\geq N$ and~$M\neq N$},\\
    M &\gg N\quad\mbox{if~$m_{ij}>n_{ij}$ for all~$i$ and~$j$},
\end{align*}
with corresponding respective conventions for~$\leq$,~$<$ and~$\ll$. We say that a vector~$v\in\mR^{\tn}$ is nonnegative, positive, or strictly positive if~$v \geq 0$,~$v >0$ or~$v \gg 0$, respectively, noting also that the symbol~$\ggg$ is used in the literature for~$\gg$. We use the same terminology for matrices. The symbols~$\mR^{\tn}_+$ and~$\mR^{\tn \times \texttt{m}}_+$ denote the set of nonnegative vectors and matrices, respectively.

The superscript~${}^{\top}$ denotes both matrix and vector transposition. 

As usual, the identity matrix is denoted by~$I$, the size of which shall be consistent with the context. 

Let~$M\in\mR^{\tn \times \tn}$ denote a square matrix. Recall that~$M$ is called Metzler if all its off-diagonal entries are nonnegative, and reducible if there exists a permutation matrix~$P$ such that
\[
    PMP^\top = 
    \bpm{
    *&0 \\ * & *
    },
\]
where~$*$ denotes an unimportant matrix block. The matrix~$M$ is said to be irreducible if~$M$ is not reducible. We let $\rho(M)$ denote the spectral radius of $M$, and $M$ is called Hurwitz (sometimes also asymptotically stable) if every eigenvalue of $M$ has negative real part. Hurwitz matrices are invertible.

%
%
We shall make extensive use of weighted one-norms. For~$\tn \in \mN$ and a vector~$z\in\mR^{\tn}$ with~$i$-th component~$z_i$, we let~$|z|$ denote the vector with~$i$-th component equal to~$|z_i|$, with the corresponding convention for matrices.  For~$v \in \mR^{\tn}_+$, define the semi-norm~$\lvert \cdot \rvert_v$ on~$\mR^\tn$ by
\begin{equation*}\label{p norm}
    |z|_v := v^{\top}|z| \quad \forall \: z\in\mR^{\tn}. 
\end{equation*}
Observe that~$|z|_v := v^{\top}z$ for all~$z\in\mR^{\tn}_+$. It is routine to show that~$\lvert \cdot \rvert_v$ is a norm if, and only if,~$v$ is strictly positive.
The terminology weighted one-norms stems from the equalities
\[
    \| z \|_1 := \sum_{i =1}^{\tn} |z_i| = \lvert z \rvert_\mathbbm{1} \quad \forall \: z \in \mR^\tn\,, 
\]
where~$\| \cdot \|_1$ denotes the usual Euclidean one-norm on~$\mR^{\tn}$ and the symbol~$\mathbbm{1}$ denotes the vector with every component equal to one.

\subsection{Function spaces}\label{sec:function_spaces}

As usual, for~$1 \leq s \leq \infty$, we let~$L^s(\mR_+,\mR^{\tn})$ denote the normed space of (equivalence classes of) measurable~$s$-integrable ($s$ finite), or essentially bounded ($s$ infinite), functions~$f : \mR_+ \rightarrow \mR^{\tn}$. We let~$L^s_{\rm loc}(\mR_+,\mR^{\tn})$ denote the usual localised version.
For~$\alpha \in \mR$  we define the weighted~$L^s$ space
\[ L^s_\alpha(\mR_+,\mR^\tn) := \big\{ f \in L^s_{\rm loc}(\mR_+, \mR^\tn) \: : \:  \| f \|_{L^s_\alpha(\mR_+)} = \| t \mapsto \e^{\alpha t} f(t) \|_{L^s(\mR_+)} < \infty\big \}\,,\]
which is a Banach space when equipped with the norm~$\| \cdot \|_{L^s_\alpha(\mR_+)}$.

For nonnegative~$r \in \mR^\tn_+$, and~$0 \leq t_0 \leq t_1 \leq \infty$, we write
\begin{equation}\label{eq:Lp_norms_weighted}
    \| f \|_{L^s(t_0,t_1 \mc r)} := \Big(\int_{t_0}^{t_1} \lvert f(\theta) \rvert_r^s \, \rd \theta\Big)^{\frac{1}{s}} \quad 1 \leq s < \infty \quad \text{and} \quad \| f \|_{L^\infty(t_0,t_1 \mc r)} := \esssup_{\theta \in [t_0, t_1]} \lvert f(\theta) \rvert_r\,,
\end{equation} 
and
%
\[\| f \|_{L^s_\alpha(t_0,t_1\mc r)} := \| t \mapsto \e^{\alpha t} f(t) \|_{L^s(t_0,t_1 \mc r)} \quad \forall \: f \in L^s_{\rm loc}(\mR_+, \mR^\tn)\,.\]
These are semi-norms in general, and are norms when~$r$ is strictly positive and, by norm equivalence, 
are equivalent to the usual~$L^s$- and~$L^s_\alpha$-norms on the interval~$(t_0,t_1)$, respectively. Finally, we write
\[ \| f \|_{L^s_\alpha(\mR_+ \mc r)} := \| f \|_{L^s_\alpha(0,\infty\mc r)} \quad \forall \: f \in L^s(\mR_+,\mR^\tn)\,.\]
Let~$R = \mR_+$ or~$\mR$. For~$\tau \in R$ we denote by~$\sigma_{\tau} : L^1_{\rm loc}(R, \mR^\tn) \to L^1_{\rm loc}(R,\mR^\tn)$ the translation operator by~$\tau$, that is,~$(\sigma_\tau f)(t) = f(t+ \tau)$ for all~$t \in R$. This is a left-shift when~$\tau \geq 0$. With a slight abuse of notation to avoid a proliferation of symbols, we use the same symbol for all such translation operators, that is, independently of the Euclidean space in~$L^1_{\rm loc}(R,\mR^\tn)$.

Table~\ref{tab:notation} summarises the symbolic quantities used, which we have sought to use consistently, and indicates instances where the notation appears as well. The table is intended to  assist with interpretation of notation, and we do not claim that the table is exhaustive.
\begin{table}[h!]
    \centering
    	\bgroup
	\def\arraystretch{1.2}%
{\footnotesize
    \begin{tabular}{|C{2.75cm}|L{8cm}|L{4.5cm}|} \hline
       {\bfseries Notation } & {\bfseries Description} & {\bfseries Illustrative equations}  \\ \hline
      ~$t$,~$t_0$,~$t_1$,~$\theta$,~$\tau$ & times (function arguments) & Throughout \\
      ~$(u,w)$,~$x$,~$y$ & input, state, output variables & \eqref{eq:lti} or~\eqref{eq:lure_4b}\\
      ~$\tm$,~$\tn$,~$\tp$ & input-, state- and output-space dimensions & \eqref{eq:lti} or~\eqref{eq:lure_4b}\\
      ~$A$,~$B$,~$C$,~$D$,~$\bG$ & matrix data of control systems, transfer function & \eqref{eq:lti} or~\eqref{eq:lure_4b} \\
       ~$p$,~$q$,~$r$,~$l$,~$k$ & vectors arising in linear dissipation inequalities & \eqref{eq:KYP_equations} or~\eqref{eq:linear_dissipative} \\
      ~$\alpha$,~$\xi$ & positive constants, rates, weightings & \eqref{eq:linear_dissipative}, \eqref{eq:incremental_main_H2} or~\eqref{eq:UaLsloc}\\
      ~$f$,~$\Delta$ & nonlinear function, linear bound term & \eqref{eq:lure_4b} or hypothesis~\ref{A2} \\
      ~$s$,~$s_0$,~$s_1$ & integrability constants & \eqref{eq:Lp_norms_weighted}, \eqref{eq:incremental_corollary_H1} or~\eqref{eq:UaLsloc}\\ 
      ~$\sigma_\tau$ & translation operator & Section~\ref{sec:function_spaces}, \eqref{eq:shift}\\
      ~$\sB$ & behaviour --- set of trajectories of Lur'e systems & Section~\ref{sec:lure_systems}\\\hline
    \end{tabular}}
    \egroup
    \caption{Commonly used notation}
    \label{tab:notation}
\end{table}

Finally, some additional function spaces are required for our treatment of almost periodic functions, and these are introduced in Section~\ref{sec:periodic}.

%
%
\subsection{A positive linear dissipation inequality}\label{sec:linear}

Whilst our overall focus of the present work relates to Lur'e systems, their analysis is underpinned by properties of the underlying linear system. Particularly relevant here are linear control systems which are positive with a certain linear dissipativity property, and in this section we gather required material.

For which purpose, consider the following controlled and observed system of linear differential equations
\begin{equation}\label{eq:lti}
    \dot x = A x + Bu, \quad y = Cx + Du\,,
\end{equation}
where~$A \in \mR^{\tn \times \tn}$,~$B \in \mR^{\tn\times \tm}$,~$C \in \mR^{\tp \times \tn}$ and~$D \in \mR^{\tp \times \tm}$ for $\tn, \tm, \tp \in \mN$. 

The linear system~\eqref{eq:lti} is called (internally) positive if~$A$ is Metzler and~$B,C,D \geq 0$. Metzler matrices go by a number of other terms in the literature, such as essentially nonnegative in~\cite[p.\ 146]{MR1019319} or quasi-positive in~\cite[p.\ 60]{MR1319817}. They are the continuous-time analogue of nonnegative matrices which arise naturally in discrete-time nonnegative dynamical
systems (see, for example,~\cite{luenberger1979introduction,MR1784150}). Further background on Metzler matrices appears in, for example,~\cite{MR1019319,varga2000matrix}.

We recall the following well-known facts related to the Metzler matrix~$M \in \mR^{n \times n}$.
\begin{enumerate}[label = (F\arabic*)]
    \item \label{ls:F1}~$\e^{Mt}$ is a nonnegative matrix for all~$t \geq 0$.
    \item \label{ls:F3}~$M$ is Hurwitz if, and only if,~$M^{-1} \leq 0$.
\end{enumerate}
For proofs of these statements, we refer to, for example~\cite[Theorem 3]{MR0277550}, 
and \cite[characterization N38 in section 6.2]{MR1298430} or \cite[characterization F15]{MR0444681}, respectively. Fact one is in fact a characterisation of Metzler matrices. 

The transfer function associated with~\eqref{eq:lti} is~$\bG(z) = C(zI-A)^{-1}B + D$ which is a proper rational function of the complex variable~$z$. For positive systems with~$A$ Hurwitz, it follows that~$\bG(0) \geq0$. 

%
%
\begin{lemma}\label{lem:dissipation}
Consider the positive linear control system~\eqref{eq:lti}, let~$\xi  \geq 0$ be given and assume that~$A + \xi  I$ is Hurwitz. Let~$(r,q) \in \mR^\tm_+ \times \mR^\tp_+$ be given. The following statements are equivalent.
\begin{enumerate}[label = {\rm(\roman*)}]
    \item \label{ls:dissipation_i} there exist~$p,l \in \mR^\tn_+$ and~$k \in \mR^\tm_+$ satisfying
\begin{equation}\label{eq:KYP_equations}
\begin{aligned}
 p^\top A + \xi  p^\top + q^\top C + l^\top & = 0\,, \\
 p^\top B + q^\top D -r^\top + k^\top & = 0\,.
\end{aligned}
\end{equation}
\item \label{ls:dissipation_ii}~$r^\top - q^\top \bG(-\xi ) \geq 0$
\end{enumerate}
If  either statement above holds and, additionally,
\begin{equation}\label{eq:positive_observable}
    q^\top C (-A)^{-1} \gg 0\,,
\end{equation}
then~$p \gg 0$.
\end{lemma}
Some remarks are in order. Observe that the condition in statement~\ref{ls:dissipation_ii} is a simple componentwise inequality of vectors which, for given~$\bG$ can always be satisfied if~$r, q$ are sufficiently (componentwise) large or small, respectively, and, is independent of positive multiplicative constants. In other words, a choice of scaling is available in the pair $(r,q)$. Lemma~\ref{lem:dissipation} is part of a larger result, which one would naturally call a~$(r,q)$-linearly dissipative lemma as it provides a state-space characterisation of the~$(r,q)$-linear dissipative property for positive linear control systems. Indeed, it can be shown that the above statements are additionally equivalent to the inequality
\[ \| y \|_{L^1(t_1,t_2 \mc q)} \leq \| u \|_{L^1(t_1,t_2 \mc r)} \quad \forall \: 0 \leq t_1 \leq t_2\,,\]
holding for all trajectories~$(u,x,y)$ of~\eqref{eq:lti} with nonnegative~$u\in L^1_{\rm loc}(\mR_+, \mR^\tm)$ and~$x(0) = 0$.

Expanded versions of the linear dissipativity lemma which essentially (but not exactly) include Lemma~\ref{lem:dissipation} appear in~\cite[Theorem 3]{ebihara20111}, ~\cite[Lemma 1]{MR3126787}, \cite[Theorems 4 and 5]{rantzer2011distributed}, and \cite[Proposition 5]{MR3356067}, broadly in the context of robust control of positive linear systems. In each of these cases~$r = \mOne$ and~$q = \gamma \mOne$ are taken as (scaled) vectors of ones, for scalar~$\gamma >0$. The authors of~\cite{MR2104542,MR2655815} refer to the~$(\mOne,\mOne)$-linear dissipative property as {\em non-accumulative}, cf.~\cite[Definition 5.4]{MR2655815}. Another linear dissipativity lemma appears in~\cite[Theorem 6.2]{MR2104542}; see also~\cite[Chapter 5]{MR2655815}, broadly with the focus of providing sufficient conditions for asymptotic stability of positive Lur'e systems, as in~\cite[Theorem 7.2]{MR2104542}.

In this sense, and as noted in~\cite[p.\ 51]{MR2104542}, the linearly dissipative lemma is a linear analogue of the celebrated positive real lemma, also termed the (K)alman-(Y)akubovich-(P)opov lemma. Linear dissipativity concepts, rather than their quadratic counterparts, are often more appropriate for describing mass conservation or balance laws which typically underpin positive systems. By way of background, we refer the reader to, for example,~\cite{anderson1973network,rantzer1996kalman} for more background on the KYP Lemma.

The condition~\eqref{eq:positive_observable} plays a role in the work~\cite{franco2021persistence} and a number of characterisations of~\eqref{eq:positive_observable} appear in~\cite[Proposition 3.2]{franco2021persistence}. Intuitively, it is a necessary and sufficient condition for the pair~$(q^\top C, A)$ to be {\em positive trajectory observable}, meaning if~$y(t) = q^\top Cx(t)$, and~$\dot x(t) = Ax(t)$ with~$x(0) \in \mR^n_+$, then the condition~$y(t)= 0$ for all~$t \geq 0$ implies that~$x(0) = 0$. The positive trajectory observability property is called {\em zero-state observable} in~\cite[Definition 5.3]{MR2104542}, and appears for discrete-time control systems in~\cite[Definition 4.4]{MR2343058}. We comment that~\eqref{eq:positive_observable} is equivalent to 
\[ \ker \sO(q^\top C, A ) \cap \mR^n_+ = \{0 \}\,,\]
where~$\sO(q^\top C, A)$ denotes the usual observability matrix of the pair~$(q^\top C, A)$, but does not enforce (usual) observability of the pair~$(q^\top C, A)$. 
%

%
%
\begin{proof}[Proof of Lemma~\ref{lem:dissipation}]
From fact~\ref{ls:F3} we have that~$(-(A+\xi  I))^{-1} \geq 0$. If statement~\ref{ls:dissipation_i} holds, then routine calculations show that
\[ p^\top = \big(q^\top C + l^\top \big)(-(A+\xi  I))^{-1}\,,\] 
and so
\begin{align*}
0&= p^\top B + q^\top D - r^\top + k^\top = q^\top \big( C (-(A+\xi  I))^{-1}B + D)  -r^\top + k^\top + l^\top (-(A+\xi  I))^{-1} B \\
& \geq q^\top \bG(-\xi ) -r^\top\,. 
\end{align*}
since~$k^\top, l^\top (-(A+\xi  I))^{-1} B \geq 0$. Rearranging the above gives statement~\ref{ls:dissipation_ii}. For the converse, we take~$p^\top := q^\top C(-(A+\xi  I))^{-1}$ which is readily shown to satisfy~\eqref{eq:KYP_equations} with~$l:=0$ and~$k := r - \bG(-\xi)^\top q \geq 0$. 

For the final claim, note that, since both~$A$ and~$A+\xi  I~$ are Metzler, and $A+\xi I$ is Hurwitz, it follows that $A$ is Hurwitz by~\cite[Corollary 4.3.2, p.\ 60]{MR1319817} (monotonicity of the spectral abscissa for Metzler matrices). Therefore, condition~\eqref{eq:positive_observable} yields
\[0 \ll q^\top C (-A)^{-1} = \int_{\mR_+} q^\top C \e^{At}\, \rd t = \int_{\mR_+} q^\top C \e^{(A+\xi  I)t} \e^{-\xi  t}\, \rd t \leq q^\top C(-(A+\xi  I))^{-1}\,.\]
Since~$p^\top \geq q^\top C(-(A+\xi  I))^{-1}$, we have~$p \gg 0$ as required.
\end{proof}

\section{Lur'e systems}\label{sec:lure_systems}

We consider the forced four-block Lur'e system~\eqref{eq:lure_4b}. Here~$A\in \mR^{\tn \times \tn}, B_i\in \mR^{\tn \times \tm_i}$ and~$C_i \in \mR^{\tp_i \times \tn}$ for~$i = 1,2$, and~$f:\mR_+ \times \mR^{\tp_1} \to \mR^{\tm_1}$ is a time-varying nonlinearity.  The variables~$x, w$ and~$y$ denote the state, forcing term and output of system~\eqref{eq:lure_4b}, taking values in~$\mR^\tn$,~$\mR^{\tm_2}$ and~$\mR^{\tp_1+\tp_2}$, all respectively.

Lur'e system~\eqref{eq:lure_4b} arises as the feedback connection of the linear control system~\eqref{eq:lti} with 
\[ B = \bpm{B_1 &B_2}, \quad C = \bpm{C_1 \\ C_2}\,,\]
and input variable~$(u,w)$, and the static nonlinear feedback~$u = f(y_1)$. As such, the two inputs~$u$ and~$w$ are a control variable and an external forcing(or disturbance)  term, respectively. The variables~$y_1$ and~$y_2$ denote the components of the measured variable~$y$ used for feedback and as a performance variable, respectively.

The transfer function~$\bG$ associated with the linear control system in~\eqref{eq:lure_4b} is given by
\[ \bG(s) := \bpm{C_1 \\ C_2}(sI-A)^{-1} \bpm{B_1 & B_2 }\,,\]
and we let~$\bG_{ij}(s) := C_i (sI - A)^{-1} B_j$ for~$i,j = 1,2$.

%
%
The following assumptions on the model data in~\eqref{eq:lure_4b} are imposed throughout. 
\begin{enumerate}[label = {\bfseries (A\arabic*)}, ref = {\rm (A\arabic*)}]
     \item \label{A1} The matrix~$A$ is Metzler, and~$B_1$,~$B_2$,~$C_1$ and~$C_2$ are nonnegative matrices.
    \item \label{A2} The  function~$f$ satisfies:
    \begin{itemize}
        \item~$t\mapsto f(t,\zeta)$ is in~$L^1_{\rm loc}(\mR_+,\mR^{\tm_1})$ for every~$\zeta \in \mR^{\tp_1}$;
        \item there exists~$\zeta_0 \in \mR^{\tp_1}$ such that~$t \mapsto f(t,\zeta_0)$ is essentially bounded, and;
\item 
    \begin{equation}\label{eq:f_incremental}
       \esssup_{t \geq 0} |f(t,\zeta_1)-f(t,\zeta_2)|\leq \Delta|\zeta_1-\zeta_2|\quad\forall \: \zeta_1,\zeta_2\in\mR^{\tp_1}\,,
    \end{equation}
    for some fixed~$\Delta\in\mR_+^{\tm_1\times \tp_1}$. 
    \end{itemize}
       \end{enumerate}
The first two items in hypothesis~\ref{A2} are mild regularity requirements, and are always satisfied if~$f$ is independent of its first variable. Observe that the inequality in~\eqref{eq:f_incremental} is a componentwise inequality of vectors in~$\mR^{\tm_1}$. A consequence of hypothesis~\ref{A2} is that~$f$ is measurable in its first variable, for each fixed second variable, and globally Lipschitz in its second variable, uniformly with respect to its first variable. We comment that so-called {\em slope restricted} nonlinearities, which are common hypothesis in many absolute stability criteria, satisfy an inequality of the form~\eqref{eq:f_incremental}. 

We shall assume throughout that the forcing term~$w : \mR_+ \rightarrow \mR^{\tm_2}$ in~\eqref{eq:lure_4b} is locally integrable, that is,~$w\in L_{\rm loc}^1(\mR_+,\mR^{\tm_2})$. Given such a~$w$, we call the pair~$(w,x)$ a trajectory of~\eqref{eq:lure_4b} if~$x$ is a locally absolutely continuous function~$x: \mR_+\rightarrow \mR^{\tn}$ which satisfies the differential equation in~\eqref{eq:lure_4b} almost everywhere. The collection of all trajectories of~\eqref{eq:lure_4b} is called the behaviour of~\eqref{eq:lure_4b} and is denoted~$\sB$. If~$f(\cdot\,,0) = 0$, then~$(0,0)$ is a constant trajectory of~\eqref{eq:lure_4b}, but we do not impose this assumption throughout. In light of the combined assumptions in~\ref{A2}, it follows from known results in the theory of ordinary differential equations (see, for example~\cite[Appendix C]{sontag2013mathematical}) that, for every~$z\in \mR^\tn$, there exists a unique trajectory~$(w,x)$ of~\eqref{eq:lure_4b} with~$x(0) = z$ so that, in particular, the set of trajectories~$\sB$ is not empty. 

It is straightforward to show that trajectories of~\eqref{eq:lure_4b} are characterised as those locally absolutely continuous functions~$x$ which satisfy the so-called variation of parameters formula
\begin{equation}\label{eq:vop}
    x(t_1+t_0) = \e^{At_1}x(t_0) + \int_{t_0}^{t_1 + t_0} \e^{A(t_1+t_0-t)} \big( B_1f(t,C_1 x(t)) + B_2 w(t)\big) \, \rd t \quad \forall \: t_1, t_0 \geq 0\,,
\end{equation}
which, note, is well defined when~$w\in L_{\rm loc}^1(\mR_+,\mR^{\tm_2})$. It is for this reason that we only impose local integrability of~$w$, rather than local essential boundedness, which is the more common assumption in ISS literature when considering general nonlinear control systems. Physically motivated forcing terms will typically belong to~$L^\infty_{\rm loc}(\mR_+, \mR^{\tm_2})$, though.

Given~$(w,x) \in \sB$, for convenience we call~$x$ and~$y$ in~\eqref{eq:lure_4b} a state trajectory and output trajectory, respectively, which we also write as~$x(\cdot \mc x_0,w)$ and~$y(\cdot \mc x_0,w)$ when~$x(0) =x_0$ for~$x_0 \in \mR^\tn$. The triplet~$(w,x,y)$ is called an  input/state/output trajectory of Lur'e system~\eqref{eq:lure_4b} if~$(w,x) \in \sB$ and~$y = y(\cdot \mc x(0),w)$. 

Observe that we do not impose that~$f(t, \cdot)$ is nonnegative valued (this allows for negative feedback by putting the minus sign on the~$f$, rather than the~$B_1$ term). However, the typical situation for positive systems is that~$f(t, \cdot)$ maps~$\mR_+^{\tp_1}$ into~$\mR_+^{\tm_1}$ and, in this case, a consequence of hypothesis~\ref{A1}, fact~\ref{ls:F1} and the variation of parameters formula~\eqref{eq:vop} with~$t_0 =0$ is that~$x(t) \geq 0$ for all~$t \geq 0$ whenever~$x(0) \geq 0$ and~$w \geq 0$ almost everywhere. This positive invariance property may also be shown by invoking, for example,~\cite[Proposition 2.2]{MR2655815}. Moreover, if~$A$ is Hurwitz, then~$\bG_{ij}(0)$ is well defined and nonnegative for all~$i,j = 1,2$.

We record the following hypotheses which will be used throughout (but not necessarily simultaneously in any of the following results).
\begin{enumerate}[label = {\bfseries (H\arabic*)}, ref = {\rm (H\arabic*)}, start = 1]
    \item \label{H1} There exist~$\xi>0$ and strictly positive~$p\in\mR^{\tn}_+$ such that 
    \[    p^\top(A+B_1\Delta C_1)\leq -\xi p^\top\,. \]
    \item \label{H2} There exist~$\xi>0$, strictly positive~$p\in\mR^{\tn}_+$, nonnegative~$l \in \mR^\tn_+$,~$q\in\mR^{\tp_2}_+$ and~$r,k\in\mR^{{\tm_2}}_+$ such that
    \begin{equation}
    \label{eq:linear_dissipative}
        \left.
        \begin{aligned}
             p^\top (A+B_1\Delta C_1) +\xi p^\top + q^\top C_2 + l^\top  & = 0\\
            p^\top  B_2 -r^\top  +k^\top  & = 0.
        \end{aligned}
        \right \}
    \end{equation}
\end{enumerate}
%
%
The following two lemmas provide characterisations of hypothesis~\ref{H1} and~\ref{H2}, respectively. To minimise the disruption to the main arguments, we relegate these proofs to the Appendix.

\begin{lemma}\label{lem:H1} Let~$A, B_1, C_1$ be as in~\ref{A1}, and fix~$\Delta \in \mR^{\tm_1 \times \tp_1}_+$.     The following statements are equivalent.
   \begin{enumerate}[label = {\rm (\roman*)}]
   \item \label{ls:H1_i} \ref{H1} holds
   \item \label{ls:H1_ii}~$A + B_1 \Sigma C_1$ is Hurwitz for all~$\Sigma \in \mR^{\tm_1 \times \tp_1}$ with~$0 \leq \Sigma \leq \Delta$.
   \item \label{ls:H1_iii}~$A$ is Hurwitz and~$\rho(\bG_{11}(0)\Delta)<1$.
    \end{enumerate}
\end{lemma}

\begin{lemma}\label{lem:H2} 
Let~$A, B_i, C_i$ for~$i = 1,2$ be as in~\ref{A1},  fix~$\Delta \in \mR^{\tm_1 \times \tp_1}_+$, and let~$r \in \mR^{\tm_2}_+$,~$q \in \mR^{\tp_2}_+$ be given. Assume that $A+B_1 \Delta C_1$ is Hurwitz and 
\[ q^\top C_2( -(A+B_1 \Delta C_1))^{-1} \gg 0\,.\]
Hypothesis~\ref{H2} is equivalent to  $A+B_1 \Delta C_1 + \xi I$ is Hurwitz and
\[ r^\top - q^\top C_2( -\xi I - (A+B_1 \Delta C_1))^{-1} B_2 \geq 0\,.\]
The previous inequality is equivalent to
\[ r^\top - q^\top \Big(\bG_{22}(-\xi)  + \bG_{21}(-\xi)\big(I - \Delta \bG_{11}(-\xi)\big)^{-1} \Delta \bG_{12}(-\xi)\Big) \geq 0\,.\]
\end{lemma}

By way of commentary, statement~\ref{ls:H1_ii} of Lemma~\ref{lem:H1} is the key stability hypothesis in the real Aizerman conjecture for positive Lur'e systems, as considered in~\cite{drummond2022aizerman}. Of course, in the single-input single-output setting, the condition~$\rho(\bG_{11}(0)\Delta)<1$ reduces to the small-gain inequality~$\bG_{11}(0)\Delta <1$. Although statement~\ref{ls:H1_iii} of Lemma~\ref{lem:H1} essentially replaces one eigenvalue condition, namely on the eigenvalues of~$A + B_1 \Delta C_1$, with another eigenvalue condition~$\rho(\Delta\bG(0)) = \rho(\bG_{11}(0)\Delta)<1$, in usual applications~$\tm_1, \tp_1 \ll \tn$ so that at least one of the second eigenvalue problems is smaller than the first.

It is clear that hypothesis~\ref{H2} implies~\ref{H1}. Assumption~\ref{H2} is a so-called linear dissipativity hypothesis, and equalities of the form~\eqref{eq:linear_dissipative} are the focus of Lemma~\ref{lem:dissipation}.  Whilst the current formulation allows for~$q,r,k,l$ to all equal zero, in practice we shall desire~$r$ and~$q$ in~\eqref{eq:linear_dissipative} to be non-zero. Hypothesis~\ref{H2} may be stated with the slack variables~$k$ and~$l$ omitted by replacing the equations in~\eqref{eq:linear_dissipative} with
 \[         
   p^\top (A+B_1\Delta C_1) +\xi p^\top + q^\top C_2 \leq 0 \quad \text{and} \quad p^\top  B_2 -r^\top  \leq 0\,.
\]
\section{Incremental stability}\label{sec:incremental_stability}

%
%
The following theorem is the main result of this work, and shows that hypotheses~\ref{H1} and~\ref{H2} are, respectively, sufficient for the Lur'e system~\eqref{eq:lure_4b} to admit two incremental bounds which we present. A number of subsequent stability and convergence properties are underpinned by this result. 
\begin{theorem}\label{thm:incremental}
Consider the four-block Lur'e system~\eqref{eq:lure_4b}, and let~$(w_a,x_a,y_a)$ and~$(w_b,x_b,y_b)$ denote two input/state/output trajectories of~\eqref{eq:lure_4b}. If~\ref{H1} is satisfied, then the estimate
\begin{equation}\label{eq:incremental_main_H1}
    |(x_{a}-x_{b})(t_1)|_p \leq \e^{-\xi (t_1 -t_0)}|(x_{a}-x_{b})(t_0)|_p +  \e^{-\xi t_1}\| w_{a}-w_{b} \|_{L^1_\xi (t_0,t_1 \mc B_2^\top p)} \quad\forall \: t_1 \geq t_0\geq 0\,,
\end{equation} 
holds. If~\ref{H2} is satisfied, then the following estimate holds:
\begin{align}
     &\int_{t_0}^{t_1} \e^{\xi t} \Big(|(y_{2,a}-y_{2,b})(t)|_q +|(w_{a}-w_{b})(t)|_k + |(x_{a}-x_{b})(t)|_l \Big) \,\rd t  + \e^{\xi t_1}|(x_{a}-x_{b})(t_1)|_p \notag \\
     & \quad \leq \e^{\xi t_0}|(x_{a}-x_{b})(t_0)|_p +  \int_{t_0}^{t_1} \e^{\xi t }|(w_{a}-w_{b})(t)|_r \,\rd t  \quad\forall \: t_1 \geq t_0 \geq 0\,. \label{eq:incremental_main_H2}
\end{align}
\end{theorem}
Here~$y_{2,k} := C_2 x_k$ for~$k = a,b$. We enumerate input/state/output trajectories of~\eqref{eq:lure_4b} using letters to reduce potential confusion between pairs of trajectories and the segments~$y_{1}$ and~$y_2$ of the output of~\eqref{eq:lure_4b}.
Observe that if~\ref{H2} holds, then the estimate~\eqref{eq:incremental_main_H1} holds with~$B_2^\top p$ replaced by~$r$, as an immediate consequence of~\eqref{eq:incremental_main_H2}. We shall use this fact in certain later arguments.

\begin{proof}[Proof of Theorem~\ref{thm:incremental}]
The derivations of~\eqref{eq:incremental_main_H1} and~\eqref{eq:incremental_main_H2} share the same start. Let~$(w_a,x_a,y_a)$ and~$(w_b,x_b,y_b)$ denote two input/state/output trajectories of~\eqref{eq:lure_4b}. The functions~$x_a$ and~$x_b$ are locally absolutely continuous, and hence so is~$t \mapsto |x_{a}(t)-x_{b}(t)|_p$ as the composition of a locally absolutely continuous function and a globally Lipschitz one. Therefore,~$t \mapsto |x_{a}(t)-x_{b}(t)|_p$ is differentiable almost everywhere on~$\mR_+$ and the fundamental theorem of Lebesgue integral calculus holds; see, for example~\cite[Theorem 3.20, p.\ 77]{leoni2017first}. In particular, for all~$0 \leq t_0 < t_1$,
\begin{align}
     \Big[ \e^{\xi t}|(x_{a}-x_{b})(t)|_p \Big]_{t_0}^{t_1}  &=   \int_{t_0}^{t_1} \dfrac{\rd }{\rd t} \left(\e^{\xi t}|(x_{a}-x_{b})(t)|_p \right ) \, \rd t \notag \\
     &=  \int_{t_0}^{t_1} \e^{\xi t} \left(\dfrac{\rd }{\rd t}|(x_{a}-x_{b})(t)|_p + \xi|(x_{a}-x_{b})(t)|_p\right)\, \rd t\,. \label{eq:incremental_p1}
\end{align}
 We proceed to derive an estimate for the derivative of~$|x_{a}-x_{b}|_p$ which appears in~\eqref{eq:incremental_p1}. For which purpose, let~$t \in (t_0,t_1)$ be a point of differentiability for~$|x_{a}-x_{b}|_p$ (and so almost every~$t \in (t_0,t_1)$ is suitable).  The variation of parameters formula~\eqref{eq:vop} yields
\begin{align}\label{egalite trajectoire}
    (x_{a}-x_{b})(t+h)  & =  \e^{Ah}(x_{a}-x_{b})(t)+ \int_{t}^{t+h}\e^{A(t+h-\theta)}B_1\left(f(C_1x_a)-f(C_1x_b)\right)\,\rd \theta \notag \\
    & \qquad  +\int_{t}^{t+h}\e^{A(t+h-\theta)}B_2 (w_a-w_b)\,\rd \theta \quad \forall \: h > 0 \,,
\end{align}
where we have suppressed the first variable of~$f$ and the argument~$\theta$ in the integrand for clarity. 
For brevity, set~$z : = x_a - x_b$. Routine estimates of~\eqref{egalite trajectoire} yield that, for all~$h >0$,
\begin{align*}
    |z(t+h)|_p &\leq  p^\top \e^{Ah}\left|z(t)\right|
  +p^\top   \int_{t}^{t+h}\e^{A(t+h-\theta)} \big(B_1\left|f(C_1x_a)-f(C_1x_b)\right| + B_2 |w_a-w_b| \big)\,\rd \theta\,.
\end{align*}
From an application of the incremental bound~\eqref{eq:f_incremental} and the nonnegativity of~$C_1$, we estimate that
\begin{align*}
|z(t+h)|_p &\leq  p^\top \e^{Ah}\left|z(t)\right|  +p^\top   \int_{t}^{t+h}\e^{A(t+h-\theta)} \big(B_1 \Delta C_1 \left|z\right| + B_2 |w_a-w_b| \big)\,\rd \theta
\end{align*}
again, for all~$h>0$. Therefore, 
\begin{align}
    \dfrac{|z(t+h)|_p-|z(t)|_p}{h} & \leq  p^\top \left(\dfrac{\e^{Ah}-I}{h}\right)\left|z(t)\right| \notag\\
    & \quad + p^\top  \int_{t}^{t+h}\e^{A(t+h-\theta)} \big(B_1 \Delta C_1 \left|z\right| + B_2 |w_a-w_b| \big)\,\rd \theta \quad \forall \: h >0\,.
    \label{eq:incremental_p2}
\end{align}
We take the limit~$h \searrow 0$ in~\eqref{eq:incremental_p2}. For which purpose, we record that
\begin{equation}\label{eq:matrix_exp_limit}
    \dfrac{\e^{Ah}-I}{h} \rightarrow A \quad \text{ as }\quad  h \searrow 0,
\end{equation}
and 
\begin{align}
    \e^{A(t+h)}\dfrac{1}{h}  \int_{t}^{t+h}\e^{-A\theta}\left(B_1\Delta C_1\left|(x_a-x_b)(\theta)\right|+B_2\left|(w_a-w_b)(\theta)\right|\right)\,\rd \theta \notag \\
    \rightarrow B_1\Delta C_1\left|z(t)\right|+B_2\left|(w_a-w_b)(t)\right| \quad \text{ as }\quad  h \searrow 0\,, \label{eq:limit_integral_h}
\end{align}
by the Fundamental Theorem of Calculus. Substituting~\eqref{eq:matrix_exp_limit} and~\eqref{eq:limit_integral_h} into~\eqref{eq:incremental_p2}, and letting~$h \searrow 0$, gives that
\begin{equation}\label{eq:incremental_p3}
    \diff{}{t}|x_{a}(t)-x_{b}(t)|_p \leq p^\top(A+B_1\Delta C_1)\left|(x_a-x_b)(t)\right|+p^\top B_2\left|(w_a-w_b)(t)\right|\,. 
\end{equation}
Since~$t \in (t_0,t_1)$ was arbitrary, we conclude that the estimate~\eqref{eq:incremental_p3} holds for almost all~$s \in (t_0, t_1)$.
Consequently, the conjunction of~\eqref{eq:incremental_p1} and~\eqref{eq:incremental_p3} yields 
\begin{align}
    \e^{\xi t_1}|(x_{a}-x_{b})(t_1)|_p - \e^{\xi t_0}|(x_{a}-x_{b})(t_0)|_p  &\leq    \int_{t_0}^{t_1} \e^{\xi t} (p^\top (A+B_1\Delta C_1)+\xi p^\top)|(x_{a}-x_{b})(t)| \, \rd t \notag \\
    & \quad +   \int_{t_0}^{t_1} \e^{\xi t} p^\top B_2|(w_{a}-w_{b})(t)| \, \rd t\,. 
    \label{eq:incremental_p_pause}
\end{align}
If hypothesis~\ref{H1} holds, then the inequality~\eqref{eq:incremental_main_H1} follows directly from~\eqref{eq:incremental_p_pause}.
If instead hypothesis~\ref{H2} holds, then we invoke this in~\eqref{eq:incremental_p_pause} to obtain
\begin{align}
    \e^{\xi t_1}|(x_{a}-x_{b})(t_1)|_p - \e^{\xi t_0}|(x_{a}-x_{b})(t_0)|_p  &\leq    \int_{t_0}^{t_1} \e^{\xi t} (-l^\top-q^\top C_2)|(x_{a}-x_{b})(t)| \, \rd t \notag \\
    & \quad +   \int_{t_0}^{t_1} \e^{\xi t} (r^\top-k^\top)|(w_{a}-w_{b})(t)| \, \rd t\,. \label{eq:incremental_p4}
\end{align}
Therefore, inequality~\eqref{eq:incremental_p4}, when combined with 
\[
    q^\top|(y_{2,a}-y_{2,b})(t)| = q^\top|C_2(x_{a}-x_{b})(t))|\leq q^\top C_2|(x_{a}-x_{b})(t)|\quad \forall \: t\geq 0\,,
\]
establishes inequality~\eqref{eq:incremental_main_H2}, as required. \qedhere
\end{proof}

%
%
The next result is our first consequence of Theorem~\ref{thm:incremental} and is an exponential incremental stability estimate for the four-block Lur'e system~\eqref{eq:lure_4b}. 
\begin{corollary}\label{cor:incremental_stability}
Consider the forced Lur'e system~\eqref{eq:lure_4b}, and let~$(w_a,x_a,y_a)$ and~$(w_b,x_b,y_b)$ denote two input/state/output trajectories of~\eqref{eq:lure_4b}. Let~$1 \leq s \leq \infty$. If~\ref{H1} is satisfied, then there exists~$\alpha_s >0$ (independent of the trajectories) such that for all~$t_1 \geq t_0 \geq 0$:
\begin{equation}\label{eq:incremental_corollary_H1}
    |(x_{a}-x_{b})(t_1)|_p \leq \e^{-\xi (t_1-t_0)}|(x_{a}-x_{b})(t_0)|_p + \alpha_s \| w_{a}-w_{b} \|_{L^s(t_0,t_1 \mc B_2^\top p)}\,. 
\end{equation}
If~\ref{H2} is satisfied, then for all~$t_1 \geq t_0 \geq 0$:
\begin{align}
\| y_{2,a} - y_{2,b} \|_{L^1_\xi(t_0,t_1 \mc q)} + \| w_a - w_b \|_{L^1_\xi(t_0,t_1 \mc k)} + \| x_{a} - x_{b} \|_{L^1_\xi(t_0,t_1 \mc l)} + \e^{\xi t_1}|(x_{a}-x_{b})(t_1)|_p \notag \\ 
\leq \e^{\xi t_0}|(x_{a}-x_{b})(t_0)|_p + \| w_{a} - w_{b} \|_{L^1_\xi(t_0,t_1 \mc r)} 
\,. \label{eq:incremental_corollary_H2}
\end{align}
\end{corollary}

%
%
The inequality~\eqref{eq:incremental_corollary_H1} with~$s = \infty$ is an incremental exponential input-to-state-stability estimate for~\eqref{eq:lure_4b}, whilst estimate~\eqref{eq:incremental_corollary_H2} entails incremental (weighted)~$L^1$-input/output stability of system~\eqref{eq:lure_4b} with a linear (unity) gain,  that is,
\[ \| y_{2,a} - y_{2,b} \|_{L^1_\xi(t_0,t_1 \mc q)} \leq \| w_{a} - w_{b} \|_{L^1_\xi(t_0,t_1 \mc r)}\,.\]
Observe that $r$ and $q$ in hypothesis~\ref{H2} appear as weights in the above estimate, and that $q \gg0$ is required for the left-hand side of the above to contain a norm of $y_{2,a} - y_{2,b}$. (The right-hand side may always be bounded above by some norm in the case that~$r \not \gg0$.) The conclusions of Corollary~\ref{cor:incremental_stability} illustrate the different consequences of hypotheses~\ref{H1} and~\ref{H2}. Both of these hypotheses contain an ``internal stability'' component --- namely that~$A + B_1 \Delta C_1$ is Hurwitz. In the context of positive Lur'e systems, this assumption is necessary and sufficient for Lyapunov stability notions (see, for example~\cite{drummond2022aizerman}). Hypothesis~\ref{H2} ensures the incremental input/output estimates~\eqref{eq:incremental_corollary_H2} and that above from~$w_a - w_b$ to~$y_{2,a}-y_{2,b}$ in terms of the linear data arising in~\eqref{eq:linear_dissipative}. Hypothesis~\ref{H1} by itself does seemingly not entail any such input-output estimate, other than by directly estimating~$y_{2,a} - y_{2,b} = C_2 (x_a - x_b)$ and invoking~\eqref{eq:incremental_corollary_H1}. 

%
%
\begin{proof}[Proof of Corollary~\ref{cor:incremental_stability}]
Inequality~\eqref{eq:incremental_corollary_H1} follows from~\eqref{eq:incremental_main_H1}, the cases~$s=1$ and~$s=\infty$ being trivial, indeed,~$\alpha_1 = 1$ and~$\alpha_\infty = 1/\xi$. For notational simplicity set $\tilde r: = B_2^\top p$.  For finite~$s \in (1,\infty)$, let~$1 < s_0 < \infty$ be conjugate to~$s$, that is,~$1/s + 1/s_0 = 1$. Then, for all~$t_1 \geq t_0 \geq 0$, an application of H\"{o}lder's inequality yields
\begin{align*}
    \e^{-\xi t_1}   \int_{t_0}^{t_1} \e^{\xi t}|(w_{a}-w_{b})(t)|_{\tilde r} \,\rd t & \leq  \e^{-\xi t_1}  \Big(\int_{t_0}^{t_1} \e^{s_0\xi t} \, \rd t\Big)^{\frac{1}{s_0}}  \Big(\int_{t_0}^{t_1} \lvert (w_a - w_b)(t) \rvert_{\tilde r}^s \, \rd t\Big)^{\frac{1}{s}} \\
    & \leq \e^{-\xi t_1} \Big( \frac{\e^{s_0 \xi t_1}}{s_0 \xi}\Big)^{\frac{1}{s_0}} \| w_{a}-w_{b} \|_{L^s(t_0,t_1\mc \tilde r)} = 
    (s_0\xi)^{-\frac{1}{s_0}} \| w_{a}-w_{b} \|_{L^s(t_0,t_1\mc \tilde r)}\,,
\end{align*}
that is, the desired second term on the right-hand side of~\eqref{eq:incremental_corollary_H1}.

The inequality~\eqref{eq:incremental_corollary_H2} is simply a rewriting of~\eqref{eq:incremental_main_H2} in terms of weighted~$L^1$-norms, and follows immediately from Theorem~\ref{thm:incremental}.
\end{proof}

We conclude this section with commentary on variations of our results.
%
%
\begin{remark}
(i) Let~$K \in \mR^{\tm_1 \times \tp_1}$. The so-called loop shifted Lur'e system associated with~\eqref{eq:lure_4b} is given by
\[
\begin{aligned} 
        \dot{x}(t) &= (A+B_1 K C_1) x(t)+B_1\big(f(t,C_1x(t)) - K C_1 x(t)\big)+B_2w(t),\\
         y(t)&= 
         \bpm{
            y_1(t)\\ y_2(t)
        }
        = 
        \bpm{
            C_1\\ C_2
        }
        x(t)\,,
\end{aligned}
\]
 which is itself an instance of~\eqref{eq:lure_4b} with~$A$ and~$f$ replaced by~$A_K : = A + B_1K C_1$ and~$f_K(t,z) : = f(t,z) - K z$, respectively. The present results are applicable in this setting when~$A_K$ and~$f_K$ satisfy assumptions~\ref{A1} and~\ref{A2}.

(ii) Consider next the following Lur'e system with output disturbance, that is, where the differential equation in~\eqref{eq:lure_4b} is replaced by
\[ \dot{x}(t) = Ax(t)+B_1f(t,C_1x(t) + w_1(t) )+B_2w_2(t)\,. \]
By introducing an error term, we rewrite the above as the differential equation in~\eqref{eq:lure_4b} with~$B_2 w$ there replaced by
\[ \tilde w (t) :=  B_2w_2(t) + B_1\big(f(t,C_1x(t) + w_1(t) ) - f(t,C_1x(t))\big)\,, \]
and apply the current results to the trajectory~$(\tilde w, x)$ of~\eqref{eq:lure_4b} with $B_2$ replaced by $I$. In light of the incremental sector condition in~\ref{A2}, observe that the componentwise absolute value~$\lvert \tilde w(t) \rvert$ may be bounded from above in terms of~$\lvert w_1(t) \rvert$ and~$\lvert w_2(t) \rvert$, independently of~$x$.

(iii) Now suppose that hypothesis~\ref{A2} is replaced by, for given non-empty subsets $Z_1, Z_2 \subseteq \mR^{\tp_1}$,
\begin{enumerate}[label = {\bfseries (A\arabic*)'}, ref = {\rm (A\arabic*)'}, start = 2]
    \item \label{A2'} The  function~$f$ satisfies:
    \begin{itemize}
        \item there exist $a_1, a_2 >0$ such that
        \[ \esssup_{t \geq 0}\| f(t,\zeta)\| \leq a_1 \| \zeta \| + a_2 \quad \forall \: \zeta \in \mR^{\tp_1}\,;\]
\item there exists~$\Delta\in\mR_+^{\tm_1\times \tp_1}$ such that 
    \begin{equation}\label{eq:f_incremental_Y}
       \esssup_{t \geq 0} |f(t,\zeta_1)-f(t,\zeta_2)|\leq \Delta|\zeta_1-\zeta_2|\quad\forall \: y_1 \in Z_1, \; \forall \: \zeta_2\in Z_2\,.
    \end{equation}
    \end{itemize}
\end{enumerate}
The first item in hypothesis~\ref{A2'} is imposed to ensure existence of trajectories of~\eqref{eq:lure_4b} (in particular, existence of state variables defined on all of $\mR_+$) and could be replaced by other suitable conditions. Note that $\Delta$ in hypothesis~\ref{A2'} may, in general, depend on the subsets $Z_1, Z_2 \subseteq \mR^{\tp_1}$.

Then, the conclusions of Theorem~\ref{thm:incremental} and Corollary~\ref{cor:incremental_stability} remain valid, only now for input/state/output trajectories~$(w_a,x_a,y_a)$ and~$(w_b,x_b,y_b)$ with $y_a \in Z_1$ and $y_b \in Z_2$ (inclusions understood pointwise almost everywhere). The only changes to the proofs are that references to~\eqref{eq:f_incremental} are replaced by references to~\eqref{eq:f_incremental_Y}. Observe that Theorem~\ref{thm:incremental} and Corollary~\ref{cor:incremental_stability} are recovered when $Z_1 = Z_2 = \mR^{\tp_1}$. One special case of interest is when $Z_2$ is a singleton, and $(w_b,x_b,y_b)$ is equal to a constant trajectory with $y_b \in Z_2$. In summary, on the one hand, the inequality~\eqref{eq:f_incremental_Y} is much less restrictive than~\eqref{eq:f_incremental}, and permits ``localized'' versions of our results. On the other hand, the results are only valid in the somewhat restrictive case that $y_a \in Z_1$ and $y_b \in Z_2$, which will in general require prior knowledge of the input/state/output trajectories to verify.
\end{remark}

%
%
\section{Consequences of incremental stability}\label{sec:consequences}

In the following two subsections we collect consequences of the incremental stability results of Section~\ref{sec:incremental_stability}. Unless otherwise specified, here we assume that
\[ f(t,z) = f(z) \quad \text{that is,~$f$ does not depend explicitly on~$t$.} \]
Since we focus on qualitative properties of the response of the state~$x$ to the external forcing term, from which analogous properties of the output~$y = Cx$ readily follow, hypothesis~\ref{H1} shall be sufficient for the majority of the remaining results.

\subsection{Convergence properties}\label{sec:convergence}

Here we investigate the state response of the forced Lur'e system~\eqref{eq:lure_4b} to forcing terms which are convergent, with main result Proposition~\ref{prop:convergence} below. For which purpose, we gather some technical material relating to the existence of constant trajectories of~\eqref{eq:lure_4b}.  To minimise disruption to the present section, the proof of the next lemma appears in the Appendix. 

For notational simplicity, we shall from hereon in make a slight abuse of notation and associate to each $z \in \mR^\tn$ the constant function $\mR_+ \to \mR^\tn$ with value $z$.

%
%
\begin{lemma}\label{lem:CICS_technical}
Consider the forced Lur'e system~\eqref{eq:lure_4b} and assume that~\ref{H1} holds. If there exist strictly positive~$v \in \mR^{\tp_1}$ and~$\rho \in (0,1)$ such that
    \begin{equation}\label{eq:G11_pf}
    v^{\top}\bG_{11}(0)\Delta \leq \rho v^{\top}\,,
\end{equation}
then, for every~$w_*\in \mR^{ \mathrm{\tm_2}}$, there exists a unique~$x_*\in\mR^\tn$ such that 
\begin{equation}
    \label{eq:x*_constant}
    0 = Ax_*+B_1f(C_1x_*)+B_2w_*\,.
\end{equation}
In particular,~$(w_*,x_*) \in \sB$. If~$f$ maps~$\mR_+^{\tp_1}$ into~$\mR_+^{\tm_1}$ and~$w_* \geq 0$, then~$x_* \geq 0$.
\end{lemma}
When the nonlinearity~$f$ in~\eqref{eq:lure_4b} is scalar valued (that is,~$\tm_1 = \tp_1 = 1$), then condition~\eqref{eq:G11_pf} simplifies to the scalar inequality~$\bG_{11}(0) \Delta <1$ (the positive constant~$v$ cancels from both sides of~\eqref{eq:G11_pf}), which is a consequence of Lemma~\ref{lem:H1} as hypothesis~\ref{H1} holds. More generally, a simple sufficient condition for~\eqref{eq:G11_pf} is that~$\bG_{11}(0)\Delta$ is irreducible, in which case~\eqref{eq:G11_pf} holds with equality for some strictly positive~$v$ and~$\rho = \rho(\bG_{11}(0)\Delta) <1~$ by the Perron-Frobenius theorem and Lemma~\ref{lem:H1}.

For the next result, recall that a subset~$V \subseteq L^\infty(\mR_+, \mR^{\tm_2})$ is said to be {\em equi-convergent} to~$v_* \in \mR^{\tm_2}$ if, for all~$\eps >0$, there exists~$\tau \geq 0$ such that
\[ \| \sigma_\tau v - v_* \|_{L^\infty} \leq \eps \quad \forall \: v \in V\,.\]
%
%
\begin{proposition}\label{prop:convergence}
    Consider the Lur'e system~\eqref{eq:lure_4b},  and assume that~\ref{H1} is satisfied. Let~$1 \leq s < \infty$, and let~$(w_a,x_a), (w_b,x_b) \in \sB$. The following statements hold.
    \begin{enumerate}[label = {\rm (\roman*)}]
        \item  If~$w_a(t)-w_b(t) \rightarrow 0$ as~$t\rightarrow \infty$, or~$w_a-w_b \in L^s(\mR_+,\mR^{\tm_2})$, then~$x_a(t)-x_b(t) \rightarrow 0$  as~$t\rightarrow \infty$.
\item If~$w_a - w_b \in L^s_\gamma(\mR_+,\mR^{\tm_2})$ for some~$\gamma >0$, then~$x_a(t)-x_b(t) \rightarrow 0$ exponentially as~$t\rightarrow \infty$.
        \item \label{ls:convergence_equilibrium} Assume that~\eqref{eq:G11_pf} holds. If~$w_a(t)\rightarrow w_*$ as~$t\rightarrow\infty$, then~$x_a(t)\to x_*$ as~$t \to \infty$. Moreover, if~$V$ is equi-convergent to~$v_*$, then for every~$L >0$, the set of solutions
        \[ \big \{ x(\cdot\,; x^0, v) \: : \: (x^0 , v) \in \mR^\tn_+ \times V \; \text{with} \; \| x^0 \| + \| v\|_{L^\infty} \leq L \big\}\,,\]
        is equi-convergent to~$x_*$.
    \end{enumerate} 
\end{proposition}

%
%
\begin{proof}
Let~$(w_a,x_a), (w_b,x_b) \in \sB$ denote two trajectories of~\eqref{eq:lure_4b}. Define~$\tilde r : = B_2^\top p \in \mR^{\tm_2}_+$.

(i) Define~$h : \mR_+ \to \mR_+$ by
\[    h(t) := \e^{-\xi t}  \int_{0}^{t} \e^{\xi \theta}|(w_{a}-w_{b})(\theta)|_{\tilde r} \,\rd \theta\,.\]
 If~$w_a(t)-w_b(t) \rightarrow 0$ as~$t\to \infty$, then a standard argument writing~$h$ as
\[ h(t) = \e^{-\xi t} \Big( \int_{0}^{t/2} \e^{\xi \theta}|(w_{a}-w_{b})(\theta)|_{\tilde r} \,\rd \theta+ \int_{t/2}^{t} \e^{\xi \theta}|(w_{a}-w_{b})(\theta)|_{\tilde r} \,\rd \theta\Big) \quad \forall \: t \geq 0\,,  \]
and estimating both terms above separately shows that~$h(t) \to 0$ as~$t \to \infty$. That~$(x_a - x_b)(t) \to 0$ as~$t\to \infty$ now follows from inequality~\eqref{eq:incremental_main_H1}. If instead~$w_a-w_b \in L^s(\mR_+)$, then the estimate~\eqref{eq:incremental_corollary_H1} shows that~$x_a - x_b$ is bounded. Therefore, another application of the estimate~\eqref{eq:incremental_corollary_H1}, now with~$t_0 = t/2$, gives
\begin{align}
    |(x_{a}-x_{b})(t)|_p & \leq \e^{-\xi t/2}|(x_{a}-x_{b})(t/2)|_p + \alpha_s \| w_{a}-w_{b} \|_{L^s(t/2,t \mc {\tilde r})} \label{eq:convergence_p1} \\
    & \leq \e^{-\xi t/2} \| x_{a}-x_{b} \|_{L^\infty(0,\infty\mc p)}+ \alpha_s \| w_{a}-w_{b} \|_{L^s(t/2,\infty \mc {\tilde r})} \to 0 \quad \text{as~$t \to \infty$.} \notag
\end{align}  
(ii) Since~$\| w_{a}-w_{b} \|_{L^s(t_0,t\mc {\tilde r})} \leq \e^{-\gamma t_0} \| w_{a}-w_{b} \|_{L^s_\gamma(t_0,t\mc {\tilde r})}$ for all~$t_0 \geq 0$, it follows from~\eqref{eq:incremental_corollary_H1} that~$x_a - x_b$ is bounded. Moreover, the inequality~\eqref{eq:convergence_p1} now yields that
    \[ |(x_{a}-x_{b})(t)|_p  \leq \e^{-\xi t/2} \| x_{a}-x_{b} \|_{L^\infty(0,\infty\mc p)} + \alpha_s \e^{-\gamma t/2} \| w_{a}-w_{b} \|_{L^s_\gamma(t/2,t\mc {\tilde r})} \quad \forall \: t \geq 0\,, \]
    giving the claimed exponential convergence.

(iii)  An application of Lemma~\ref{lem:CICS_technical} ensures the existence of a constant trajectory~$(w_*,x_*)$ of~\eqref{eq:lure_4b}. The claimed convergence now follows from statement (i) with~$(w_b, x_b) = (w_*,x_*)$. The claimed equi-convergence follows from~\eqref{eq:incremental_main_H1} as, under the given hypotheses, for all~$\eps, L >0$, there exists~$\tau \geq 0$ such that~$h(t) \leq \eps = \eps(L)$ for all~$t \geq \tau$ and all~$w_a \in V$ with~$\| w_a \|_{L^\infty} \leq L$.
\end{proof}

%
%
 \subsection{Response to almost periodic forcing}\label{sec:periodic}



Here we demonstrate that a consequence of incremental stability is desirable behaviour of the state and output of the forced Lur'e system~\eqref{eq:lure_4b} with respect to almost periodic forcing signals. The results of this section are inspired by those appearing in~\cite[Sections 3 and 4]{gilmore2020infinite} and~\cite[Section 4]{gilmore2021incremental}, which themselves trace their roots back to arguments used in~\cite{angeli2002lyapunov}. The key difference is that here the crucial incremental stability properties are ensured by the positivity hypotheses and structure.

We introduce additional notation and terminology required for the results of this section. Recall that~$R=\mR_+$ or~$\mR$. We define~$BC(R,\mR^\tn)$ and~$BUC(R,\mR^\tn)$ as
the spaces of all, respectively, bounded continuous and bounded uniformly
continuous functions~$R \to \mR^\tn$. Endowed with the supremum norm,~$BC(R,\mR^\tn)$ and~$BUC(R,\mR^\tn)$ are Banach spaces.

Let~$r \in \mR^\tn$ be strictly positive. We define the space of uniformly locally~$s$-integrable functions~$UL^s_{\rm loc}(R,\mR^\tn)$ by
 \[
 UL^s_{\rm loc}(R,\mR^\tn):=\left\{w\in L^s_{\rm loc}(R,\mR^\tn) \: : \: \sup_{a\in R}\int_{a}^{a+1}\|w(t)\|^s \, \rd t<\infty\right\}\,,
 \]
where~$1\leq s<\infty$. It is straightforward to show that~$UL^s_{\rm loc}(R,\mR^\tn)$ when equipped with the Stepanov norm
\[
\|w\|_{S^s_r}:=\sup_{a\in R}\left(\int_{a}^{a+1}\lvert w(t) \rvert^s_r\, \rd t\right)^{1/s}\,,
\]
 is a Banach space. For notational simplicity, we write~$\| \cdot \|_{S^s} = \| \cdot \|_{S^s_r}$ when the value of~$r$ is unimportant. The choice of~$1$ in the upper limit of the above integral is also unimportant insomuch as replacing~$a+1$ by~$a+b$ for positive~$b$ gives rise to equivalent norms, as does the choice of strictly positive~$r \in \mR^\tn$. 
From H\"{o}lder's inequality it is routine to establish that, for~$1\leq s_0\leq s_1<\infty$,
\begin{equation}\label{eq:embedding}
UL^{s_1}_{\rm loc}(R,\mR^\tn)\subset UL^{s_0}_{\rm loc}(R,\mR^\tn)
\quad\mbox{and}\quad
\| v\|_{S^{s_0}} \leq\| v\|_{S^{s_1}} \quad \forall\:v \in UL^{s_1}_{\rm loc}(R,\mR^\tn)\,,
\end{equation}
so that the space~$UL^{s_1}_{\rm loc}(R,\mR^\tn)$ is continuously embedded in~$UL^{s_0}_{\rm loc}(R,\mR^\tn)$ for~$1\leq s_0\leq s_1<\infty$. Furthermore, for~$\alpha\geq 0$, we define
\begin{equation}\label{eq:UaLsloc}
 U_\alpha L^s_{\rm loc}(R,\mR^\tn):=\left\{w\in UL^s_{\rm loc}(R,\mR^\tn): \lim_{t\to\infty}\|\sigma_t(\exp_\alpha\,w)\|_{S^s}=0\right\}.
\end{equation}
It is clear that~$L^s_\alpha(\mR_+,\mR^\tn)\subset U_\alpha L^s_{\rm loc}(R,\mR^\tn)$, and, if~$\alpha>0$,
then~$U_\alpha L^s_{\rm loc}(R,\mR^\tn)\subset L^s_\beta(R,\mR^\tn)$ for all~$\beta\in (0,\alpha)$.

%
%
Equipped with the above notation, we present a further consequence of Theorem~\ref{thm:incremental}.
\begin{corollary}\label{cor:incremental_stability_Sp}
    Consider the Lur'e system~\eqref{eq:lure_4b}, let~$1 \leq s < \infty$, and let~$(w_a,x_a), (w_b,x_b) \in \sB$. The following statements hold.
\begin{enumerate}[label = {\rm (\roman*)}]
    \item If~\ref{H1} holds, then with~$\tilde r : = B_2^\top p$, there exist positive constants~$\beta_0, \xi_0$ such that, for all~$t_1 \geq t_0 \geq 0$
    \begin{equation}\label{eq:incremental_corollary_Sp}
     |(x_{a}-x_{b})(t_1)|_p \leq \beta_0 \big( \e^{-\xi_0 (t_1-t_0)}|(x_{a}-x_{b})(t_0)|_p + \| \sigma_{ t_0}(w_{a}-w_{b}) \|_{S^s_{\tilde r}} \big)\,. 
\end{equation}
If~$w_a - w_b \in U_\alpha L^s_{\rm loc}(\mR_+, \mR^{\tm_2})$, then~$(x_{a}-x_{b})(t) \to 0$ exponentially as~$t \to \infty$.
\item If~\ref{H2} holds, then for all~$t_0 \geq 0$
\begin{align}
\| \sigma_{t_0} (y_{2,a} - y_{2,b}) \|_{S^1_q}  \leq \beta_0 |(x_{a}-x_{b})(t_0)|_p + (\beta_0 + \e^\xi) \| \sigma_{t_0} (w_{a} - w_{b}) \|_{S^1_r}\,, \label{eq:incremental_corollary_S1_both}
\end{align}
where~$(w_a,x_a,y_a), (w_b,x_b,y_b)$ are two input/state/output trajectories of~\eqref{eq:lure_4b}.
    \end{enumerate}
    The constants $\beta_0, \xi_0, \beta_1$ are independent of the trajectories.
\end{corollary}
%
%
\begin{proof}
The inequality~\eqref{eq:incremental_corollary_Sp} is derived from~\eqref{eq:incremental_corollary_H1} in the same manner as the proof of~\cite[Theorem 3.4, statement (3)]{gilmore2020infinite}. The claimed exponential convergence follows from~\eqref{eq:incremental_corollary_Sp} and the inequality
\[ \e^{\alpha t} \| \sigma_{t} (w_a -w_b) \|_{S^s} \leq \big\| \sigma_{t} \big({\rm exp}_\alpha (w_a -w_b)\big) \big \|_{S^s} \quad \forall \: t \geq 0\,.\]
We, therefore, omit the details.

We proceed to establish~\eqref{eq:incremental_corollary_S1_both}. 
An application of the inequality~\eqref{eq:incremental_corollary_H2} with~$t_0 = \tau$ and~$t_1 = \tau+1$ to the trajectories~$(\sigma_{t_0} w_k, \sigma_{t_0} x_k,\sigma_{t_0} y_k)$ for~$t_0 \geq 0$ yields, for all~$\tau \geq 0$,
\begin{align}
\| \sigma_{t_0} (y_{2,a} - y_{2,b}) \|_{L^1(\tau,\tau+1 \mc q)} & \leq  \e^{-\xi \tau} \| \sigma_{t_0} (y_{2,a} - y_{2,b}) \|_{L^1_\xi(\tau,\tau+1 \mc q)} 
\notag\\
& \leq |\sigma_{t_0}(x_{a}-x_{b})(\tau)|_p  + \e^{-\xi \tau} \| \sigma_{t_0}(w_{a} - w_{b}) \|_{L^1_\xi(\tau, \tau+1 \mc r)} \notag \\
& \leq |\sigma_{t_0}(x_{a}-x_{b})(\tau)|_p  + \e^{\xi} \| \sigma_{t_0}(w_{a} - w_{b}) \|_{L^1(\tau, \tau+1 \mc r)} \,.\label{eq:incremental_main_io'}
\end{align}
Taking the supremum in~\eqref{eq:incremental_main_io'} over~$\tau \geq 0$ yields that
\begin{equation}\label{eq:incremental_corollary_S1_both_p1}
     \|\sigma_{t_0} (y_{2,a} - y_{2,b}) \|_{S^1_q}  \leq \big(\sup_{\tau \geq 0} |\sigma_{t_0}(x_{a}-x_{b})(\tau)|_p + \e^\xi \| \sigma_{t_0}(w_{a}-w_{b})\|_{S^1_r}\big)\,.
\end{equation}
An application of~\eqref{eq:incremental_corollary_Sp} with~$s = 1$ and $t_1 = t_0 + \tau$ gives that
\[ |\sigma_{t_0}(x_{a}-x_{b})(\tau)|_p \leq \beta_0 \big( \e^{-\xi_0 \tau}|\sigma_{t_0}(x_{a}-x_{b})(0)|_p + \| \sigma_{ t_0}(w_{a}-w_{b}) \|_{S^1_r} \big) \quad \forall \: \tau \geq 0\,, \]
(here using that~$\tilde r$ may be replaced by~$r$) which, when substituted into~\eqref{eq:incremental_corollary_S1_both_p1}, gives the desired bound~\eqref{eq:incremental_corollary_S1_both}.
\end{proof}
An inspection of the above proof of the inequality~\eqref{eq:incremental_corollary_S1_both} reveals that twice the right-hand side of~\eqref{eq:incremental_corollary_S1_both} is an upper bound for
\[ \| \sigma_{t_0} (y_{2,a} - y_{2,b}) \|_{S^1_q}  + \| \sigma_{t_0} ( x_{a} - x_{b}) \|_{S^1_l} \,.\]
However, note that~$l \not \gg0$ in general, so $\| \cdot \|_{S^1_l}$ is not a norm.

Our present focus is on so-called Stepanov almost periodic functions. For which purpose, we first collect material on almost periodic functions in the sense of Bohr. For further background reading on almost periodic functions, we refer the reader to the texts~\cite{MR0275061,MR2460203,MR0020163}.

Recall that~$R=\mR$ or~$\mR_+$. A set~$S \subseteq R$ is said to be {\em relatively dense} (in~$R$) if there exists~$l>0$
such that
\[
 [a,a+l]\cap S \neq \emptyset \quad \forall \: a \in R\,.
\]
For~$\eps>0$, we say that~$\tau \in R$ is an~$\eps${\em -period} of
$v\in C(R,\mR^\tn)$ if    
\[
\|v(t)-v(t+\tau)\| \leq \eps \quad \forall \: t \in R\,.
\]
We denote by~$P(v,\eps)\subseteq R$ the set of~$\eps$-periods of~$v$ and we say that~$v\in C(R,\mR^\tn)$ is {\it almost periodic} (in the sense of Bohr) if~$P(v,\eps)$ is relatively dense in~$R$ for every~$\eps>0$. We denote the set of almost periodic functions~$v\in C(R,\mR^\tn)$ by~$AP(R,\mR^\tn)$, and mention that~$AP(R,\mR^\tn)$ is a closed subspace of~$BUC(R,\mR^\tn)$. It is clear that any continuous periodic  function is almost periodic.

The readily established equality
\begin{equation}\label{eq:ap_inf_tail}
\sup_{t\geq\tau}\|v(t)\| =\|v\|_{L^\infty(R)} \quad \forall \: \tau \in R, \; \forall \: v \in AP(R,\mR^\tn)\,,
\end{equation}
shows that functions in~$AP(R,\mR^\tn)$ are completely determined by their ``infinite right tails''.

Let~$C_0(\mR_+,\mR^\tn)$ denote the subspace of~$C(\mR_+,\mR^\tn)$ of functions which converge to zero as~$t \to \infty$. We define 
\[
 AAP(\mR_+,\mR^\tn):=AP(\mR_+,\mR^\tn)+C_0(\mR_+,\mR^\tn)\,,
 \]
 as the space of all {\em asymptotically almost  periodic} functions. It is a closed subspace of~$BUC(\mR_+,\mR^\tn)$ and, since by~\eqref{eq:ap_inf_tail} we have that~$AP(\mR_+,\mR^\tn) \cap C_0(\mR_+,\mR^\tn) = \{0\}$, it follows that~$AAP(\mR_+,\mR^\tn) =AP(\mR_+,\mR^\tn) \oplus C_0(\mR_+,\mR^\tn)$ (direct sum). Consequently, for each~$v \in AAP(\mR_+,\mR^\tn)$, there exists a unique~$v^{\rm ap} \in AP(\mR_+,\mR^\tn)$ such that~$v - v^{\rm ap} \in C_0(\mR_+,\mR^\tn)$ (cf. \cite[Lemma 5.1]{gilmore2020infinite}).

The spaces~$AP(\mR_+,\mR^\tn)$ and~$AP(\mR,\mR^\tn)$ are closely related, as we now briefly recall. Indeed, it is shown in~\cite[Section 4]{gilmore2020infinite}\footnote{Itself following an idea in~\cite[Remark on p.~318]{MR0773063}} that map~$AP(\mR_+,\mR^\tn) \to AP(\mR,\mR^\tn)$,~$v \mapsto v_{\rm e}$ given by
\begin{equation}\label{eq:ap_extension}
v_{\rm e}(t):=\lim_{k \to \infty}v(t+\tau_k) \quad \forall \: t \in \mR\,,    
\end{equation}
where~$\tau_k \in P(v,1/k)$ for each~$k \in \mN$ and~$\tau_k\to\infty$ as~$k\to\infty$ has the following properties:
\begin{itemize}
    \item is well defined, that is,~$v_{\rm e} \in AP(\mR,\mR^\tn)$\,;
    \item~$v_{\rm e}$ extends~$v$ to~$\mR$\,;
    \item~$v \mapsto v_{\rm e}$ is an isometric isomorphism~$AP(\mR_+,\mR^\tn) \to AP(\mR,\mR^\tn)$.
\end{itemize}
We now recall the concept of Stepanov almost periodicity which is weaker than that of Bohr.  Let~$v\in L_{\rm loc}^s(R,\mR^\tn)$,
where~$1\leq s <\infty$, and~$\eps>0$. We say that\footnote{The terminology {\em~$\eps$-period of~$v$ in the sense of Stepanov} is used in~\cite{gilmore2020infinite} and analogously for the following concepts.}~$\tau\in R$ is a {\em Stepanov~$\eps$-period of~$v$},   if
\[ \|(\sigma_\tau-I)v\|_{S^s} = \sup_{a\in R}\left(\int_a^{a+1}\|v(\theta+\tau)-v(\theta)\|^s \rd \theta\right)^{1/s}\leq\eps\,.\]
The set of Stepanov~$\eps$-periods of~$v$ is denoted by~$P_s(v,\eps)$. The function~$v$ is called {\em Stepanov almost periodic} if, for every~$\eps>0$, the set~$P_s(v,\eps)$ is relatively dense in~$R$. 

We let~$S^s(R,\mR^\tn)$ denote the set of all functions in~$L_{\rm loc}^s(R,\mR^\tn)$ which are Stepanov almost periodic. It is clear that~$AP(R,\mR^\tn)\subset S^s(R,\mR^\tn)$ (where the inclusion is strict, since~$S^s(R,\mR^\tn)$ contains discontinuous functions), and~$P(v,\eps)\subset P_s(v,\eps)$ for every~$v\in AP(R,\mR^\tn)$ and every~$\eps>0$. Moreover, it is readily shown that~$S^s(R,\mR^\tn)$ is a closed subspace of~$UL_{\rm loc}^s(R,\mR^\tn)$ with respect to the Stepanov norm~$\|\,\cdot\,\|_{S^s}$. It follows from \eqref{eq:embedding} that, for~$1\leq s_0\leq s<\infty$ and~$v\in L^s_{\rm loc}(R,\mR^\tn)$,~$P_s(v,\eps)\subset P_{s_0}(v, \eps)$, and consequently,~$S^s(R,\mR^\tn)\subset S^{s_0}(R,\mR^\tn)$ with continuous embedding.

Routine modifications to the map~$v \mapsto v_{\rm e}$ given by~\eqref{eq:ap_extension} shows that it extends to an isometric isomorphism of~$S^s(\mR_+,\mR^\tn) \to S^s(\mR,\mR^\tn)$.

The space of Stepanov asymptotically almost periodic functions~$AS^s(\mR_+,\mR^\tn)$ is defined as
\begin{equation}\label{eq:AS_definition}
 AS^s(\mR_+,\mR^\tn):=S^s(\mR_+,\mR^\tn)+U_0L^s_{\rm loc}(\mR_+,\mR^\tn)\,.   
\end{equation}
It is shown in~\cite[Lemma 4.2]{gilmore2020infinite} that~$\|\sigma_\tau v\|_{S^s} = \| v \|_{S^s}$ for all~$v \in S^s(\mR_+,\mR^\tn)$ and all~$\tau \in \mR_+$ which, in particular, yields that~$S^s(\mR_+,\mR^\tn) \cap U_0L^s_{\rm loc}(\mR_+,\mR^\tn) = \{0\}$. Therefore, the summation in~\eqref{eq:AS_definition} is a direct sum and, consequently, for every~$v \in AS^s(\mR_+,\mR^\tn)$ there exists a unique~$v^{\rm s} \in S^s(\mR_+,\mR^\tn)$ such that~$v - v^{\rm s} \in U_0L^s_{\rm loc}(\mR_+,\mR^\tn)$, cf.~\cite[Lemma 4.4]{gilmore2020infinite}.

Finally, we let~$\sB\sB$ denote the bilateral behaviour of~\eqref{eq:lure_4b}, that is, the set of pairs~$(w,x)$ where~$w \in L^1_{\rm loc}(\mR,\mR^{\tm_2})$ and~$x : \mR \to \mR^{\tn}$ is a locally absolutely continuous function such that~$w$ and~$x$ satisfy the differential equation in~\eqref{eq:lure_4b} almost everywhere on~$\mR$. Since~$f$ is assumed to be time-independent, both $\sB$ and~$\sB\sB$ are shift-invariant, that is,
\begin{equation}\label{eq:shift}
(v,x)\in \sB \, (\sB\sB) \implies (\sigma_{\tau}v,\sigma_{\tau}x)\in \sB \, (\sB\sB)\quad \forall \,
\tau \in \mR_+ \, (\mR)\,.
\end{equation}

%
%
The following proposition is the main result of this section, and is inspired by the combination of~\cite[Theorem 4.5]{gilmore2020infinite} and~\cite[Theorems 4.3,4.5]{gilmore2021incremental}.

\begin{proposition}\label{prop:periodic}
Consider the forced Lur's sytem~\eqref{eq:lure_4b}, assume that~\ref{H1} and~\eqref{eq:G11_pf} are satisfied, and let~$ w \in S^1(\mR_+,\mR^{\tm_2})$. Then there exists a unique~$z^{\rm ap} \in AP(\mR_+,\mR^\tn)$ such that~$(w, z^{\rm ap}) \in \sB$ and, for every~$\eps >0$, there exists~$\delta>0$ such that~$P_1(w,\delta)\subset P(z^{\rm ap},\eps)$. 

Let~$(v,x) \in \sB$. The following further statements hold.
\begin{enumerate}[label = {\rm (\arabic*)}, itemsep = 1ex, topsep = 1ex, leftmargin = 1ex, itemindent = 4ex]
\item\label{ls:ap_1} If~$v\in AS^1(\mR_+,\mR^{\tm_2})$ with~$v^{\rm s}=w$, then
  \[
    \lim_{t\to\infty}\big(x(t)-z^{\rm ap}(t)\big)=0\,,
  \]
that is,~$x\in AAP(\mR_+,\mR^\tn)$ with~$x^{\rm ap}=z^{\rm ap}$.
\item \label{ls:ap_2} Statement~\ref{ls:ap_1} remains true if any of the following hold:
\begin{enumerate}[label = {\rm(\roman*)}]
    \item~$w\in AP(\mR_+,\mR^{\tm_2})$ and~$v \in AAP(\mR^+, \mR^{\tm_2})$ satisfies~$v^{\rm ap} = w$ and, in this case,~$v^{\rm ap} = v^{\rm s} = w$;
    \item~$v, w \in L^\infty(\mR_+, \mR^{\tm_2})$ and~$\| v - w \|_{L^\infty(t, \infty)}\to 0$ as~$t\to \infty$, or;
    \item~$v - w \in L^1(\mR_+, \mR^{\tm_2})$.
\end{enumerate}
\item\label{ls:ap_3} If~$v-w\in U_\alpha L^1_{\rm loc}(\mR_+,\mR^{\tm_2})$ for some~$\alpha>0$, then the convergence in statement~\ref{ls:ap_1} is exponentially fast.
\item\label{ls:ap_4} If~$w$ is periodic with period~$\tau$, then~$z^{\rm ap}$ is~$\tau$-periodic.
\item\label{ls:ap_5} $(w_{\rm e},z^{\rm ap}_{\rm e})\in\sB\sB$ and there is no other bounded function~$\hat x : \mR \to \mR^\tn$ such that~$(w_{\rm e}, \hat x) \in \sB\sB$.
\item \label{ls:ap_6} If~$(v,x^{\rm ap}) \in \sB$ with~$v \in S^1(\mR_+, \mR^{\tm_2})$ and~$x^{\rm ap} \in AP(\mR_+,\mR^\tn)$, then there exists~$\beta_0 >0$ such that
\begin{equation}\label{eq:ap_norm_diff_1}
    \| x^{\rm ap} - z^{\rm ap} \|_{L^\infty(\mR_+ \mc p)} \leq \beta_0 \| v - w\|_{S^1_{\tilde r}} \,, 
\end{equation} 
where~$p$ and~$\tilde r := B_2^\top p$ are as in~\ref{H1}. The constant~$\beta_0$ is independent of the trajectories. If, additionally,~$v, w \in L^\infty(\mR_+, \mR^{\tm_2})$, then
\begin{equation}\label{eq:ap_norm_diff_2}
 \| x^{\rm ap} - z^{\rm ap} \|_{L^\infty(\mR_+ \mc p)} \leq c_1\| v - w\|_{L^\infty(\mR_+ \mc r)} \,,
 \end{equation} 
 for some constant~$c_1 >0$ independent of the trajectories.
%
%
\item \label{ls:ap_7} If~\ref{H2} holds and~$(w_k,z_k,y_k)$ are input/state/output trajectories of~\eqref{eq:lure_4b} for~$k =a,b$ with~$w_k \in S^1(\mR_+, \mR^{\tm_2})$,~$z_k = z^{\rm ap}_k \in AP(\mR_+,\mR^\tn)$, then
 \[   \| y_{2,a} - y_{2,b} \|_{S^1_q} \leq c_2\| w_a - w_b\|_{S^1_{r}} \,,\]
for some constant~$c_2>0$ independent of the trajectories.
\end{enumerate}

\end{proposition}
Observe that (asymptotic) almost periodic properties of~$y$ in~\eqref{eq:lure_4b} follow immediately from those of~$x$, and so do not appear separately in the statement of the result above, apart from statement~\ref{ls:ap_7} which imposes hypothesis~\ref{H2}. The hypothesis that~\eqref{eq:G11_pf} holds is made to ensure that every trajectory $(w,z) \in \sB$ with $w \in S^1(\mR_+,\mR^{\tm_2})$ has bounded $z$. It could be replaced by the hypothesis that for each $w \in S^1(\mR_+,\mR^{\tm_2})$, there exists $(w,z) \in \sB$ with bounded $z$, or simply that $f(0) =0$.

The conclusions of Proposition~\ref{prop:periodic} appear, in slightly adapted form, across~\cite[Theorem 4.5]{gilmore2020infinite} and~\cite[Theorems 4.3, 4.5]{gilmore2021incremental}, with the comment that the first part of statement~\ref{ls:ap_6} and statement~\ref{ls:ap_7} appear new. So, in that sense, the bulk of the conclusions of above result are not new. Rather, the novelty of the current result is that the underlying and essential incremental stability properties come from Theorem~\ref{thm:incremental} and its corollaries, which leverage the positive systems structure considered presently, where the key differences are seen in the multi-input multi-output (MIMO) case. Indeed, the works~\cite{gilmore2020infinite, gilmore2021incremental} impose norm assumptions in the linear stabilisability and ``nonlinear ball'' hypotheses (the latter of which can be reformulated as sector conditions in terms of inner products). The assumptions~\ref{A2} and~\ref{H1}/\ref{H2} are more in the spirit of linear dissipation and linear sector conditions, which are not equivalent to norm conditions in general in the MIMO case. For more detailed descriptions of the differences between these types of assumptions, we refer the reader to~\cite[Section III. E.]{drummond2022aizerman}.

Extensive commentary is given in~\cite{gilmore2020infinite} and~\cite{gilmore2021incremental} on the results~\cite[Theorem 4.5]{gilmore2020infinite} and~\cite[Theorems 4.3, 4.5]{gilmore2021incremental}, respectively, and much of which is applicable here. We comment, for instance, that statement~(3) follows as a consequence of~\cite[Theorem 4.4]{natarajan2013behavior} which considers certain infinite-dimensional Lur'e systems. Additional results in the literature related to Proposition~\ref{prop:periodic} include~\cite[Theorem 3.2.9]{gilmore_thesis}, \cite[Theorem 4.3]{gilmore2020stability}, \cite[Theorem 2]{MR1772240} and~\cite[Theorem 1]{MR0167681}. The papers~\cite{MR1772240,MR0167681} are restricted to scalar nonlinearities, that is,~$\tm_1 = \tp_1 = 1$, and the papers~\cite{gilmore2020stability,MR1772240,MR0167681} consider almost periodic forcing functions in the sense of Bohr only. The works~\cite{gilmore2020stability,gilmore_thesis,MR0167681} and~\cite{gilmore2020infinite,MR1772240} adopt Lyapunov and input-output approaches, respectively.

A further conclusion of Proposition~\ref{prop:periodic} is valid, namely that~$\mod(z^{\rm ap}_{\rm e}) \subseteq \mod(\widetilde{w_{\rm e}})$ where~$\mod(z)$ for~$z \in AP(\mR,\mR^\tn)$ denotes the~$\mZ$-module of generalised Fourier coefficients associated with~$z$. Here~$\widetilde{z}$ denotes the so-called Bochner transform of~$z \in S^1(\mR,\mR^{\tn})$, which gives rise to a vector-valued almost periodic function in the sense of Bohr, see~\cite{gilmore2020infinite,gilmore2021incremental}. However, it follows that~$\mod(\widetilde{w_{\rm e}}) = \mod(w_{\rm e})$ if~$w \in AP(\mR_+,\mR^{\tm_2})$. For brevity, we have not reproduced the details here.

\begin{proof}[Proof of Proposition~\ref{prop:periodic}]  
The proofs of the first set of claims, and statements~\ref{ls:ap_1}, \ref{ls:ap_3}, \ref{ls:ap_4}, and~\ref{ls:ap_5}, are somewhat lengthy but essentially follow the corresponding statements of~\cite[Theorem 4.5]{gilmore2020infinite}, {\em mutatis mutandis}. There the forcing function is assumed to belong to~$S^2$, and estimates involving the~$S^2$ norm are used. Here we use the~$S^1$ norm, in particular invoking the~$L^1$ and~$S^1$ estimates in Corollaries~\ref{cor:incremental_stability} and~\ref{cor:incremental_stability_Sp}, respectively. Special attention is paid in~\cite[Theorem 4.5]{gilmore2020infinite} to the output variables, as these functions need not in general admit the same regularity properties as the state variable in the infinite-dimensional setting considered there. The setting here is much simpler in that regard. We have omitted the details for brevity.

Statement~\ref{ls:ap_2} follows from statement~\ref{ls:ap_1} as, in each case, the hypotheses imply that~$v \in AS^1(\mR_+,\mR^{\tm_2})$ and~$v^{\rm s} = w$. 

To prove statement~\ref{ls:ap_6}, first note that the existence of~$(v,x^{\rm ap}) \in \sB$ with~$x^{\rm ap} \in AP(\mR_+,\mR^\tn)$ for given~$v \in S^1(\mR_+, \mR^{\tm_2})$ follows from the first part of the proposition. Now let~$\eps >0$ be given.

It follows from the inequality~\eqref{eq:incremental_corollary_Sp} with~$s=1$, $t_0 =0$ and $t_1 = t \geq 0$, that
\[ |(x^{\rm ap} - z^{\rm ap})(t)|_p \leq \beta_0 \big( \e^{-\xi_0 t}|(x^{\rm ap} - z^{\rm ap})(0)|_p + \| v - w \|_{S^1_{\tilde r}} \big) \quad\forall \: t\geq  0\,,\]
whence there exists~$\tau = \tau(\eps)>0$ such that
\[ \sup_{t \geq \tau}\lvert (x^{\rm ap} - z^{\rm ap})(t) \rvert_p \leq \eps + \beta_0 \| v-w \|_{S^1_{\tilde r}}\,. \]
Therefore, since~$x^{\rm ap} - z^{\rm ap} \in AP(\mR_+,\mR^\tn)$, the equality~\eqref{eq:ap_inf_tail} yields that
\[ \| x^{\rm ap} - z^{\rm ap} \|_{L^\infty(\mR_+ \mc p)} \leq \eps + \beta_0 \| v - w \|_{S^1_{\tilde r}}\,.  \]
Since~$\eps >0$ was arbitrary, inequality~\eqref{eq:ap_norm_diff_1} holds.

The inequality~\eqref{eq:ap_norm_diff_2} is proven similarly, now starting from~\eqref{eq:incremental_corollary_H1} with~$s=\infty$. We leave the remaining details to the reader. The proof of statement~\ref{ls:ap_7} is also similar, and now invokes the estimates~\eqref{eq:incremental_corollary_S1_both} and then~\eqref{eq:incremental_corollary_Sp} to show that, for all~$\eps >0$, there exists~$\tau >0$ such that
\[ \| y_{2,a} - y_{2,b}\|_{S^1_q} = \| \sigma_{\tau} (y_{2,a} - y_{2,b}) \|_{S^1_q} \leq \eps + c \| \sigma_\tau( w_a - w_b) \|_{S^1_r} = \eps + c \|  w_a - w_b \|_{S^1_r}\,,\]
from which the desired estimate is inferred.
\end{proof}

\section{Examples}\label{sec:examples}

We illustrate our results with two examples.

\begin{example}\label{ex:1}
Consider~\eqref{eq:lure_4b} with $\tn = \tm_1 = \tp_1 = 3$, $\tm_2 = 2$, $\tp_2 = 1$ and model data
\begin{align*}
     A &:= \bpm{-1 & 0 & 0 \\ 0 & -10 & 0 \\ 0 & 0& -100} + 0.01 \mOne \mOne^\top, \quad B_1 := I, \quad B_2 := \bpm{0 & 0 \\ 1 & 0 \\0 &1},\\
     C_1 &:= \bpm{0 &1 &1 \\ 0 &0 &1 \\ 0 &0 &1}, \quad C_2 := \bpm{1 & 0 &0}\,,
\end{align*}
which satisfies~\ref{A1}, and so-called diagonal (or repeated) nonlinearity $f : \mR^3 \to \mR^3$ given by
\[ f(t,y) = f(y) := \bpm{g(y_1) & 0 & 0 \\ 0 & g(y_2) & 0 \\ 0 & 0 & g(y_3)} \quad \forall \: y \in \mR^3, \]
where $g : \mR\to \mR$ is a slope-restricted nonlinearity satisfying
\begin{equation}\label{eq:slope_condition}
     0 \leq \frac{g(z_1) - g(z_2)}{z_1 - z_2} \leq \delta \quad \forall \: z_1, z_2 \in \mR, \; z_1 \neq z_2\,,
\end{equation}
for some $\delta >0$. A consequence of the above assumption on $g$ is that $f$ satisfies~\ref{A2} with $\Delta = \delta I$. The Lur'e system
\[ \dot x = A x + B_1 f(C_1 x)\,,\]
is considered in~\cite[Section III. F., Example 3]{drummond2022aizerman} in the context of comparing the real Aizerman conjecture for positive Lur'e systems to other absolute stability criteria. It is shown there that hypothesis~\ref{H1} holds with $\delta \approx 89.89$. 

A calculation shows that $\Delta \bG_{11}(0) \gg0$, so that $\Delta \bG_{11}(0)$ is {\em a forteriori} irreducible, and hence hypothesis~\eqref{eq:G11_pf} holds. In particular, the hypotheses of Proposition~\ref{prop:convergence} are satisfied, and from that result we obtain convergence of the state (and hence output) for all converging inputs. 

For a numerical simulation we take
\[ g(z) := \frac{\delta_0}{2} z + \frac{\delta_0}{2} \sin(z), \quad \delta_0 := 89\,,\]
which satisfies~\eqref{eq:slope_condition}, and a graph of which is plotted in Figure~\ref{fig:ex_1_f}. In the following simulations we take the convergent forcing terms
\[ w(t) = \bpm{w^1(t) \\ w^2(t)} := k\bpm{1 + t^2 \e^{-t} \\ 1 + \frac{t}{1+t^3}}, \quad t \geq 0\,, \quad k = 3,6,9\,,\]
so that $k$ plays the role of an amplitude parameter, and which are plotted in Figure~\ref{fig:ex_1_w}. Figure~\ref{fig:ex_1_y1} plots the performance output $y_2(t) = C_2 x(t)$ against $t$ for $(w,x) \in \sB$ with $x(0) =0$, for varying $k$. In each case, convergence to some $y_*$ over time is observed. Figure~\ref{fig:ex_1_y2} plots $(y_{2,a} - y_{2,b})(t) = C_2 (x_a - x_b)(t)$ against $t$ for $(w,x_a), (w,x_b) \in \sB$ with
\[ x_a(0) := 4k \bpm{-1&1&1}^\top \quad \text{and} \quad x_b(0) := 4k \bpm{-1 & 0 &-1}^\top, \quad k = 3,6,9\,.\]
Now convergence $(y_{2,a} - y_{2,b})(t) \to 0$ is observed.
\end{example}

%
%
\begin{figure}[h!]
    \centering
    \begin{subfigure}[]{0.425\textwidth}
    \centering
    \includegraphics[width = 0.95\textwidth]{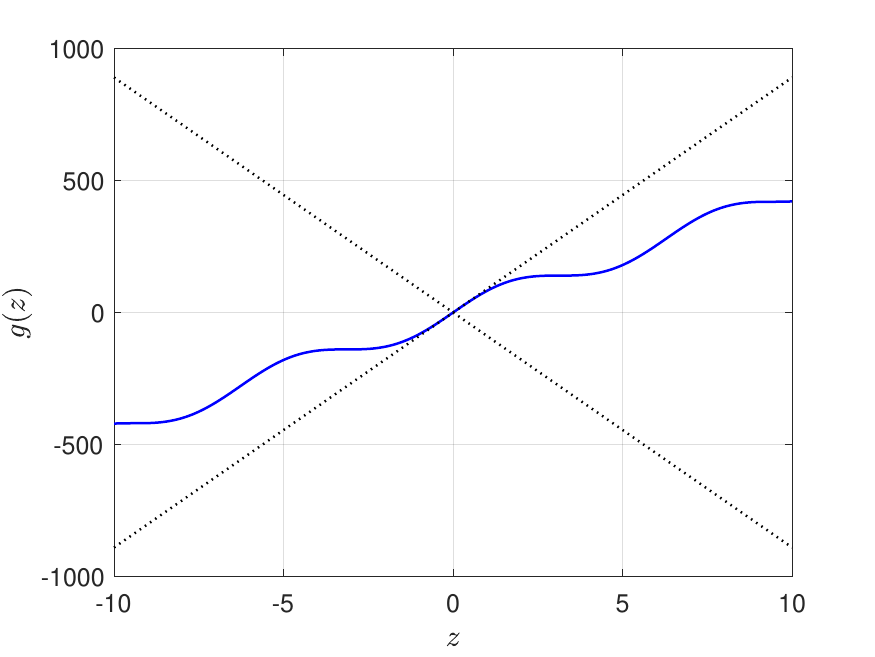}
    \caption{}
    \label{fig:ex_1_f}
    \end{subfigure}%
        \begin{subfigure}[]{0.425\textwidth}
    \centering
    \includegraphics[width = 0.95\textwidth]{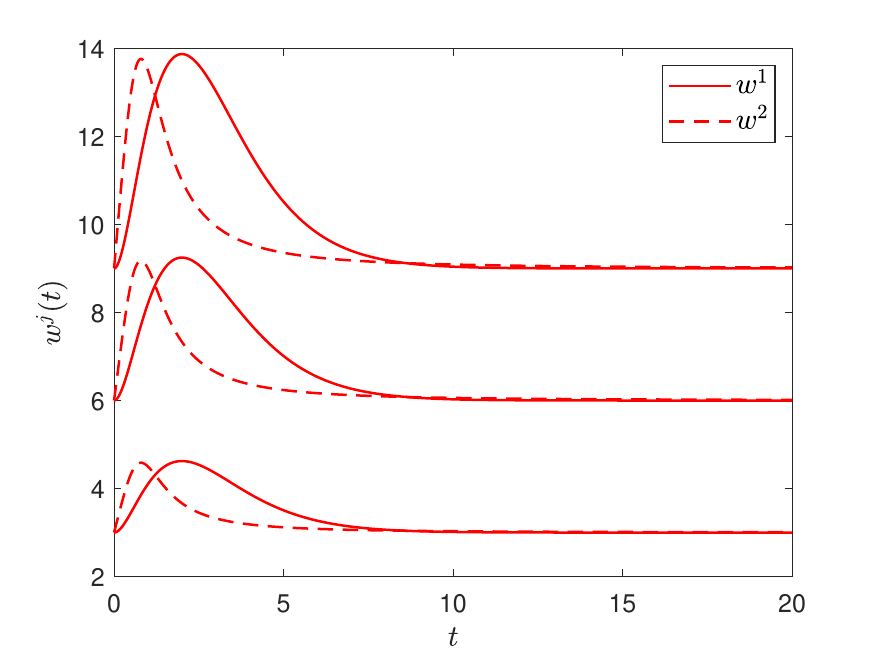}
    \caption{}
    \label{fig:ex_1_w}
    \end{subfigure}%
    \\
        \begin{subfigure}[]{0.425\textwidth}
    \centering
    \includegraphics[width = 0.95\textwidth]{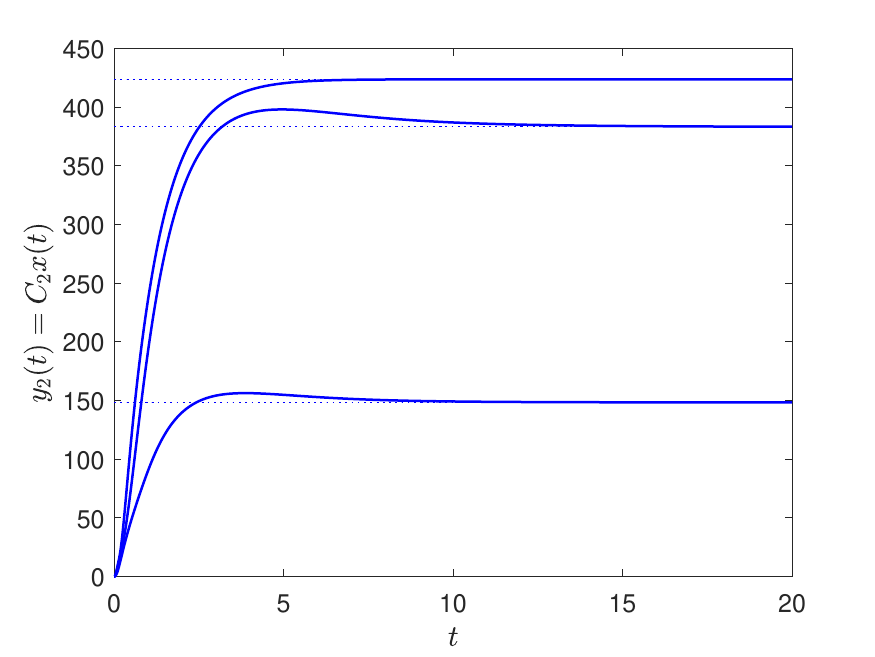}
    \caption{}
    \label{fig:ex_1_y1}
    \end{subfigure}
        \begin{subfigure}[]{0.425\textwidth}
    \centering
    \includegraphics[width = 0.95\textwidth]{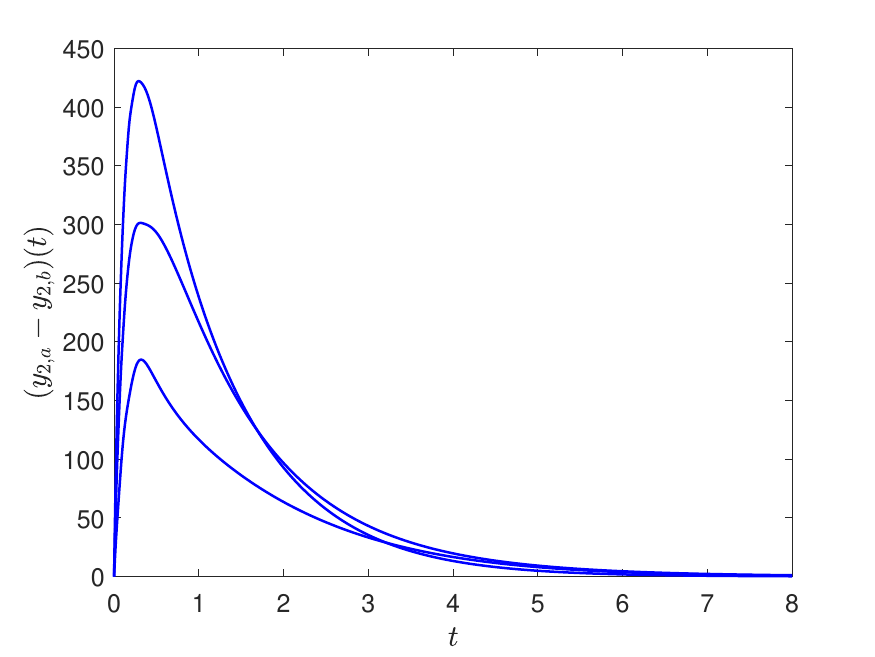}
    \caption{}
    \label{fig:ex_1_y2}
    \end{subfigure}

    \label{fig:ex_1}
    \caption{Simulation results from Example~\ref{ex:1}. See main text for a description.}
\end{figure}

\begin{example}\label{ex:2}
Consider~\eqref{eq:lure_4b} with 
\begin{align*}
     A := -(\tn+3)I, \quad B_1 := \mOne, \quad B_2 := e_1, \quad 
     C_1 := \mOne^\top , \quad C_2 := e_n^\top\,,
\end{align*}
for $\tn \in \mN$, and which satisfies~\ref{A1}. Here $e_i$ denotes the $i$-th standard basis vector in $\mR^\tn$. We take
\[ f(y) : = \frac{y}{1 + \lvert y \rvert} \quad \forall \: y \in \mR\,,\]
which is readily shown to satisfy~\ref{A2} with $\Delta :=1$. With $\tn = 4$ and $\xi: = 0.1$ we have that hypothesis~\ref{H2} is satisfied with
\begin{align*}
q & : = 1, \\
p^\top &:= -q C_2 (\xi I  + A+B_1C_1)^{-1} = \bpm{0.05  &  0.05&    0.05&   0.1949}, \\
r &:= p^\top B_2 = 0.05\,.     
\end{align*}
The hypotheses of Theorem~\ref{thm:incremental} and Corollary~\ref{cor:incremental_stability} are satisfied and so, in particular the inequality~\eqref{eq:incremental_corollary_H2} yields that, for input/state/output trajectories $(w_k,x_k,y_k)$ of~\eqref{eq:lure_4b} with $x_k(0) = 0$,
\[ \| y_{2,a} - y_{2,b} \|_{L^1_\xi(0,t \mc q)} \leq \| w_{a} - w_{b} \|_{L^1_\xi(0,t \mc r)} \quad \forall \: t \geq 0\,,\]
which here simplifies to
\[ \| x^4_{a} - x^4_b \|_{L^1_\xi(0,t)} \leq 0.05 \| w_{a} - w_{b} \|_{L^1_\xi(0,t)} \quad \forall \: t \geq 0\,.\]
As a numerical illustration of Proposition~\ref{prop:periodic}, we set
\[ w^{\rm ap}(t) : =  \sin(2\pi t) + \sin(2\sqrt{3}\pi t), \quad w^{\rm aap}(t) := 5t {\rm e}^{-t} + w^{\rm ap}(t) \quad \forall \: t \geq 0\,,\]
so that $w^{\rm ap} \in AP(\mR_+,\mR)$ and $w^{\rm aap} \in AAP(\mR_+,\mR)$. Note that $w^{\rm ap}$ is not periodic. Numerical simulations are shown in Figure~\ref{fig:ex_2_vp_a} where the performance output $x^4(t) = C_2 x(t)$ is plotted against $t$ for a different combination of initial state and forcing term. Finally, in Figure~\ref{fig:ex_2_vp_c} we plot the performance outputs $C_2 x(t;0,w^{\rm s})$ and $C_2 x(t;0,w^{\rm as})$ against $t$ where 
\begin{align*}
 w^{\rm s}(t) &:= 2+\sin\big({\rm mod}(t,3\pi/2)\big) + \sin\big(\sqrt{2}\,{\rm mod}(t,3\pi/(2\sqrt{2}))\big), &  w^{\rm as}(t) & := 5t {\rm e}^{-t/4} + w^{\rm s}(t)\,.
\end{align*}
Here, for $t\geq 0$ and $\tau >0$,
\[ {\rm mod}(t,\tau):=t-k\tau \in [0,\tau), \quad \text{where $k$ is the largest integer in $\mZ_+$ such that $t\geq k\tau$.}\]
The functions $w^{\rm s}$ and $w^{\rm as}$ are Stepanov almost periodic, and asymptotically Stepanov almost periodic, respectively. Note that $w^{\rm s} \not \in AP(\mR_+,\mR)$ as $w^{\rm s}$ is not continuous. A graph of $w^{\rm s}$ is plotted in Figure~\ref{fig:ex_2_vp_b}. Convergence of the output trajectories to one another is observed, as is convergence to an apparently almost periodic function.
\end{example}
%
%
\begin{figure}[h!]
    \centering
   \begin{subfigure}[]{0.5\textwidth}
   \centering
    \includegraphics[width = 0.95\textwidth]{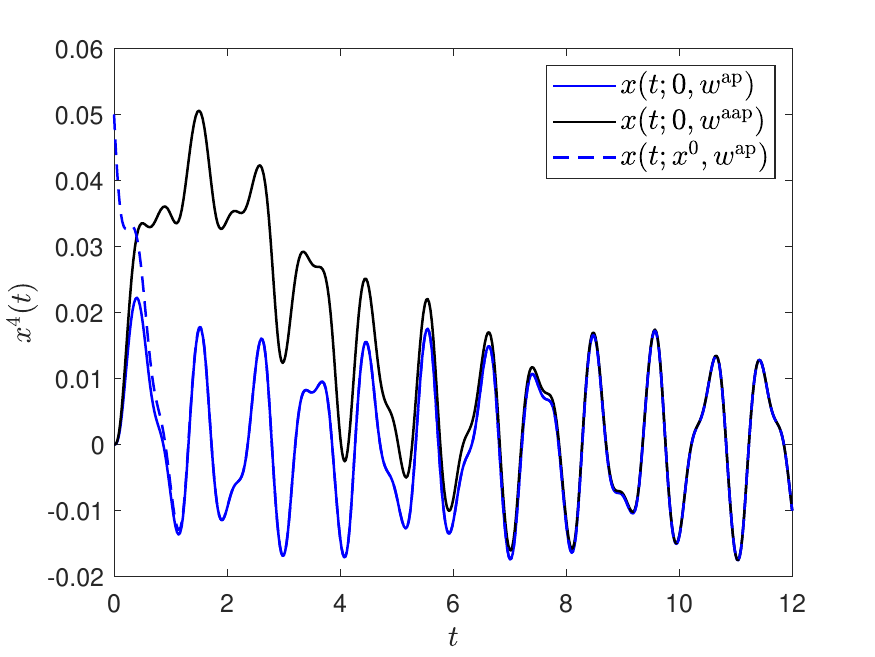}
    \caption{}
    \label{fig:ex_2_vp_a}
   \end{subfigure}%
     \begin{subfigure}[]{0.5\textwidth}
   \centering
    \includegraphics[width = 0.95\textwidth]{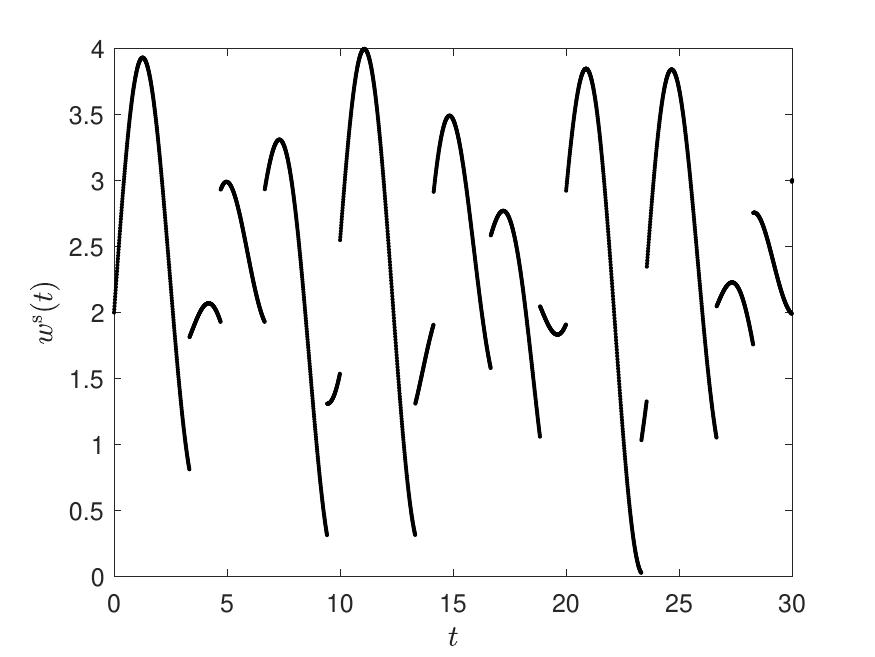}
    \caption{}
    \label{fig:ex_2_vp_b}
   \end{subfigure}\\
   \begin{subfigure}[]{0.95\textwidth}
\centering
    \includegraphics[width = 0.95\textwidth]{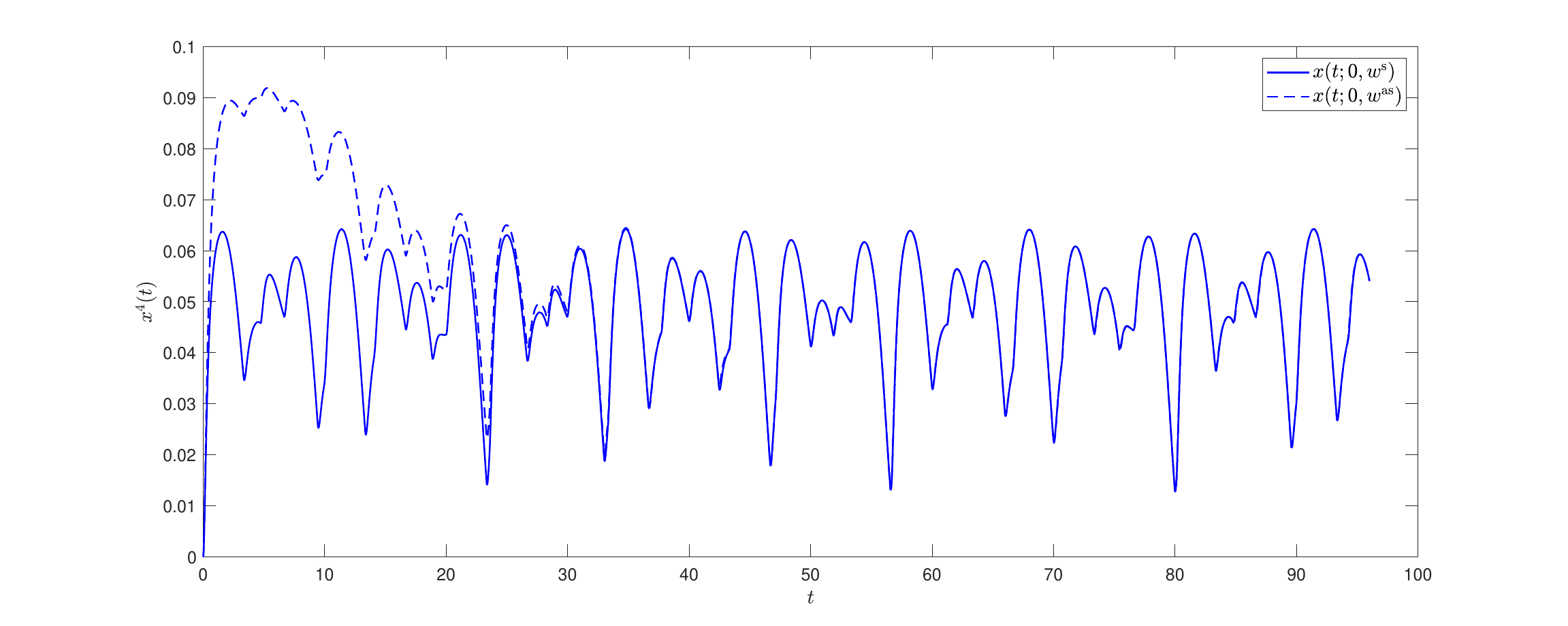}
   \caption{}
    \label{fig:ex_2_vp_c}
   \end{subfigure}
    \caption{Numerical simulations from Example~\ref{ex:2}. See main text for a description.}
    \label{fig:ex_2_vp}
\end{figure}
%
%
\section{Summary}\label{sec:summary}

Incremental stability properties have been considered for a class of systems of forced, positive Lur'e differential equations. The approach to incremental stability is via linear dissipativity theory, explored across~\cite{MR3126787,MR2104542,MR2655815,MR3356067}, which is natural when the underlying linear system is positive, and distinguishes the work from other incremental stability works by the authors, including~\cite{guiver2019infinite,gilmore2020infinite,gilmore2021incremental}. In another vein, the hypotheses of the real Aizerman conjecture for positive Lur'e systems, when combined with a linear incremental gain condition on the nonlinear component, are sufficient for various incremental stability notions. Our main incremental stability result is Theorem~\ref{thm:incremental}, which may be viewed as an ``incremental $(r,q)$-mass imbalance or dissipation'' for input/state/output trajectories of~\eqref{eq:lure_4b}.

As expected, incremental stability properties are sufficient for a range of convergence properties and (almost) periodic state and output responses to almost periodic inputs, as explored in Section~\ref{sec:consequences}. These are arguably necessary ingredients for a well-defined and tractable ``frequency response'' for forced Lur'e systems. We expect that versions of these results may be derived for systems of forced positive Lur'e difference equations, {\em mutatis mutandis}.

Finally, we comment that one motivation for the current work is the paper~\cite{revay2020lipschitz} which, roughly speaking, uses IQCs to derive incremental input-output  estimates for so-called equilibrium neural networks, which may be used to quantify their robustness. Briefly, and from a control theoretic perspective, equilibrium neural networks admit a description as the equilibrium of a forced Lur'e system with nonlinear term which satisfies a Lipschitz condition. Whilst our work makes essential use of the underlying positivity structure, should the networks considered in applications satisfy this structure, then our results become applicable.

\vspace{4ex}

{\large \bfseries Acknowledgements}

Chris Guiver's contribution to this work was supported by a Personal Research Fellowship from the Royal Society of Edinburgh. 
Chris Guiver expresses gratitude to the RSE for their support. The authors are grateful to Prof. Matthew Turner and Dr Ross Drummond for fruitful discussions related to the current work.

{\large \bfseries Statements and Declarations}

Data sharing not applicable to this article as no datasets were generated or analysed during the current study.

No generative AI was used in the production of this work.

The authors have no relevant financial or non-financial interests to disclose.

\appendix

\section{Appendix}

We provide proofs of technical lemmas not given in the main text.
%
%
\begin{proof}[Proof of Lemma~\ref{lem:H1}]
Since~$A + B_1 \Delta C_1$ is Metzler, hypothesis~\ref{H1} is equivalent to~$A + B_1 \Delta C_1$ being Hurwitz by, for example, \cite[Lemma 2.2, p.\ 31]{MR2655815}. This latter property is equivalent to statement~\ref{ls:H1_ii} by~\cite[Lemma 1]{drummond2022aizerman}.

We now prove the equivalence of statements~\ref{ls:H1_i} and~\ref{ls:H1_iii}. Assume first that~\ref{H1} holds. In particular
\[ p^\top A \leq p^\top (A+ B_1\Delta C_1) \leq -\xi p^\top\,,\]
so that~$A$ is Hurwitz and, hence, invertible. Seeking a contradiction, assume that~$\rho(\bG_{11}(0)\Delta)\geq 1$. In particular, there  exists~$\Sigma\in \mR_+^{\tm_1\times \tp_1}~$ with~$\Sigma\leq \Delta$ and nonzero~$w\in\mR_+^{\tp_1}$ such that~$\bG_{11}(0)\Sigma w = w$. Necessarily it must be the case that~$(-A)^{-1}B_1\Sigma w > 0$, as otherwise the previous equality gives the contradiction~$w=0$. 

On the one hand, we now compute that 
\[
    p^\top (A+B_1\Sigma C_1)(-A)^{-1}B_1\Sigma w = p^\top(-B_1 \Sigma w + B_1 \Sigma \bG_{11}(0)\Sigma w) = p^\top(-B_1 \Sigma w + B_1 \Sigma w) = 0,
\]
On the other hand, by hypothesis
    \[ p^\top (A+B_1\Sigma C_1)(-A)^{-1}B_1\Sigma w \leq p^\top (A+B_1\Delta C_1)(-A)^{-1}B_1\Sigma w \leq - \xi p^\top  (-A)^{-1}B_1\Sigma w <0\,,\]
    since~$p \gg 0$ and~$(-A)^{-1}B_1\Sigma w > 0$, which is a contradiction. 

Conversely, assume that~$A$ is Hurwitz and that~$\rho(\bG_{11}(0)\Delta)<1$. If~$\lambda \in \mC$ with~${\rm Re}(\lambda) \geq0$ and~$v \in \mC^\tn$ are such that
\[ (A+B_1\Delta C_1)v = \lambda v\,,\]
then~$\lambda$ is not an eigenvalue of~$A$, since~$A$ is Hurwitz, so that~$(\lambda I - A)^{-1} B_1 \Delta C_1v = v$ and, consequently,
\begin{equation}\label{eq:small_gain_MIMO_p1}
\bG_{11}(\lambda) \Delta C_1v = C_1v\,.
\end{equation} 
The positivity hypotheses in~\ref{A1} ensure that
\[ \lvert \bG_{11}(\lambda) \rvert = \Big \lvert \int_0^\infty \e^{-\lambda t} C_1 \e^{A t}B_1 \, \rd t \Big \rvert \leq \int_0^\infty \lvert \e^{-\lambda t} \rvert \cdot \lvert C_1 \e^{A t}B_1 \rvert \, \rd t \leq \int_0^\infty C_1 \e^{A t}B_1 \, \rd t = \bG_{11}(0)\,. \]
Hence, 
\[ \lvert \bG_{11}(\lambda) \Delta\rvert \leq \lvert \bG_{11}(\lambda) \rvert \cdot \lvert \Delta\rvert \leq \bG_{11}(0) \Delta\,.\]
Since~$\rho(M) \leq \rho(\lvert M \rvert)$ for all square matrices~$M$, we have that
\[ \rho(\bG_{11}(\lambda) \Delta) \leq \rho(\lvert \bG_{11}(\lambda) \Delta \rvert) \leq \rho(\bG_{11}(0) \Delta) <1\,,\]
and we conclude from~\eqref{eq:small_gain_MIMO_p1} that~$C_1v =0$. It now follows that~$v = 0$, since~$A$ is Hurwitz. 
\end{proof}

%
%
\begin{proof}[Proof of Lemma~\ref{lem:H2}]
The claim follows from an application of Lemma~\ref{lem:dissipation}. The given strictly positive hypothesis ensures that $p \gg0$. For the second condition, we invoke the  Sherman–Morrison–Woodbury formula to give
\begin{align*}
    ( -\xi I - (A+B_1 \Delta C_1))^{-1} & = \Big( (-\xi I - A)^{-1}  \\
    & \qquad + (-\xi I - A)^{-1}B_1\big( I - \Delta C_1 (-\xi I - A)^{-1} B_1 \big)^{-1} \Delta C_1 (-\xi I - A)^{-1}\Big)\,.
\end{align*}
Pre- and post-multiplying the right-hand side of the above by~$q^\top C_2$ and~$B_2$, respectively, gives the desired expression.
\end{proof}

%
%
\begin{proof}[Proof of Lemma~\ref{lem:CICS_technical}]
 Let~$w_*\in\mR^{\tm_2}$ be given. We proceed in two steps. 

{\sc Step 1: construction of~$x_*$.}
Consider the function 
\[
F : \mR^{\tp_1} \to \mR^{\tp_1}, \quad    F(y) := \bG_{11}(0)f(y)-C_1A^{-1}B_2w_* \quad \forall \: y \in \mR^{\tp_1}\,.
\]
The function~$F$ is a contraction in the~$\lvert \cdot \rvert_v$ norm, as seen from the estimates
\begin{align*}
    |F(y_1) - F(y_2)|_v &= v^{\top}|\bG_{11}(0)f(y_1) - \bG_{11}(0)f(y_2)| \notag \leq v^{\top}\bG_{11}(0)\Delta|y_1-y_2| \\
    & \leq  \rho |y_1-y_2|_v \quad \forall \: y_1, y_2\in \mR^{\tp_1} \,, 
\end{align*}
where we have invoked that~$\bG_{11}(0) \geq 0$, the incremental bound~\ref{A2}, and the inequality~\eqref{eq:G11_pf}. Therefore, by the contraction mapping theorem,~$F$ has a unique fixed point, which we denote~$y_*$. Now define
\begin{equation}
\label{eq:x*_def}
    x_* := -A^{-1}B_1f(y_*)-A^{-1}B_2w_*,
\end{equation}
and observe that
\begin{align*}
    Ax_*+B_1f(C_1x_*)+B_2w_* =& -B_1f(y_*) + B_1f(-C_1A^{-1}B_1f(y_*)-C_1A^{-1}B_2w_*)  \\
    & = -B_1f(y_*) + B_1f(F(y_*))= 0\,,
\end{align*}
since~$y_*$ is the fixed point of~$F$. We see that~$x_*$ satisfies~\eqref{eq:x*_constant}

{\sc Step 2: uniqueness of~$x_*$.} To establish uniqueness of~$x_*$ as solution of~\eqref{eq:x*_constant}, assume that~$x_\dagger \in \mR^{\tn}$ satisfies~\eqref{eq:x*_constant}. Multiplying both sides of this expression by~$C_1$ gives that
\[
    C_1x_\dagger  = -C_1A^{-1}B_1f(C_1x_\dagger )-C_1A^{-1}B_2w_*  = F(C_1x_\dagger ),
\]
and so, from the uniqueness of~$y^*$ as a fixed point of~$F$, it follows that~$C_1x_\dagger  = y^*$. Therefore, from~\eqref{eq:x*_def},
\[  x_* = -A^{-1}B_1f(y_*)-A^{-1}B_2w_* =  -A^{-1}B_1f(C_1 x_\dagger)-A^{-1}B_2w_* = x_\dagger\,,\]
as required. 

To prove the nonnegativity of~$x_*$ under the further assumptions, we in fact invoke statement~\ref{ls:convergence_equilibrium} of Proposition~\ref{prop:convergence}, which note, does not rely on the claimed positivity of~$x_* \in \mR^\tn$, so there is no circular reasoning. Taking nonnegative convergent~$w_a$, possible as~$w_* \geq 0$, and~$x_a(0) \in \mR^n_+$, it follows that~$x_a(t) \geq 0$ for all~$t \geq 0$. Therefore,~$x_* \geq 0$ as the limit as~$t \to \infty$ of~$x_a(t)$. \qedhere
\end{proof} 

%
%

\def\cprime{$'$} \def\cprime{$'$} \def\cprime{$'$}

\end{document}